%% file: main.tex
\documentclass[journal]{IEEEtran}
\usepackage{color}
\usepackage{verbatim}
\usepackage{amsfonts}
\usepackage{amssymb}
\usepackage{stfloats}
\usepackage{mathtools}
\usepackage{cite}
\usepackage{graphicx}
\usepackage{psfrag}
\usepackage{subfigure}
\usepackage{amsmath}
\usepackage{array}
\usepackage{epstopdf}
\usepackage{authblk}
\usepackage{graphicx} 
\usepackage{amsthm} 
\usepackage{lipsum}
\usepackage{verbatim} 
\usepackage{authblk}
\usepackage{mathtools}
\usepackage{cuted}
\usepackage[lined,boxed,ruled]{algorithm2e}

\usepackage{framed} 
\usepackage{subfigure}
\usepackage{soul}
\usepackage{bm}
\usepackage{setspace}
\usepackage{url}

\input{header}

\newtheorem{MR}{Main Result}
\IEEEaftertitletext{\vspace{-1\baselineskip}}

\title{Semantic-Relevance Based Sensor Selection for Edge-AI Empowered Sensing Systems}

\author{Zhiyan Liu, \emph{Graduate Student Member, IEEE}, and Kaibin Huang, \emph{Fellow, IEEE}\thanks{Z. Liu and K. Huang are with Department of Electrical and Electronic Engineering at The University of Hong Kong, Hong Kong. Contact: K. Huang  (Email: huangkb@eee.hku.hk). }}

\makeatletter
\newcommand{\removelatexerror}{\let\@latex@error\@gobble}
\makeatother

\IEEEoverridecommandlockouts

\begin{document}

\maketitle
\begin{abstract}
The \emph{sixth-generation} (6G) mobile network is envisioned to incorporate sensing and edge \emph{artificial intelligence} (AI) as two key functions. Their natural convergence leads to the emergence of \emph{Integrated Sensing and Edge AI} (ISEA), a novel paradigm enabling real-time  acquisition and understanding of sensory information at the network edge. However, ISEA faces a communication bottleneck due to the large number of sensors and the high dimensionality of sensory features. Traditional approaches to communication-efficient ISEA lack awareness of \emph{semantic relevance}, i.e., the level of relevance between sensor observations and the downstream task. To fill this gap, this paper presents a novel framework for semantic-relevance-aware sensor selection to achieve optimal \emph{end-to-end} (E2E) task performance under heterogeneous sensor relevance and channel states. E2E sensing accuracy analysis is provided to characterize the sensing task performance in terms of selected sensors' relevance scores and channel states. Building on the results, the sensor-selection problem for accuracy maximization is formulated as an integer program and solved through a tight approximation of the objective. The optimal solution exhibits a priority-based structure, which ranks sensors based on a priority indicator combining relevance scores and channel states and selects top-ranked sensors. Low-complexity algorithms are then developed to determine the optimal numbers of selected sensors and features. Experimental results on both synthetic and real datasets show substantial accuracy gain achieved by the proposed selection scheme compared to existing benchmarks.
\end{abstract}


\section{Introduction}

The envisioned \emph{sixth-generation} (6G) mobile networks are expected to encompass two innovative functions: sensing and edge \emph{artificial intelligence} (AI) \cite{ITUR2023}. The sensing function coordinates a large number of on-device sensors to achieve accurate collective perception of the physical world\cite{LetaiefCollabSense,DingzhuISCC}. In parallel, edge AI involves the widespread deployment of pre-trained AI models at the network edge, delivering low-latency intelligent services to mobile users\cite{Letaief2022JSAC}. The natural convergence of these functions gives rise to an emerging paradigm known as \emph{Integrated Sensing and Edge AI} (ISEA). ISEA provides a powerful platform for enabling a broad range of \emph{Internet-of-Things} (IoT) applications, such as smart cities, autonomous driving, cloud-controlled robotics, and digital twins\cite{liu2025integratedsensingedgeai}. However, a key design challenge in ISEA is the data uploading from numerous distributed sensors to a server for remote inference, which creates a communication bottleneck. Researchers have identified that the fundamental limitation of traditional design approaches lies in their lack of awareness of the relevance level between the semantic content of sensing data and the sensing task, termed semantic relevance\cite{LYC2020CVPR}. For example, in surveillance, sensors that do not detect humans provide no useful data for remote identification and tracking of intruders. Consequently, the traditional brute-force approach of uploading all sensing data clutters the air interface and degrades sensing performance due to irrelevant inputs. Conversely, communication overhead can be dramatically reduced without compromising sensing performance if only those sensors detecting relevant events, such as human intruders, are requested to upload their data. This highlights the need for sensor selection based on the relevance of their data to the target task. In this work, we present a novel framework for semantic-relevance aware sensor selection in an ISEA system. This framework systematically leverages the semantic matching of sensors to the target task to achieve high communication efficiency.

An ISEA system is typically implemented on the architecture of a \emph{Multi-View Convolutional Neural Network} (MVCNN)\cite{Hang2015ICCV}, also known as split inference\cite{Chen2023arxiv,Shao2022JSAC,Niu2019Infocom,Zhou2020IoTJ,Lan2023TWC}. This architecture consists of an edge server wirelessly connected to a cluster of distributed sensors that provide different observations of the same event or environment. To reduce communication overhead, a low-complexity \emph{neural network} (NN) model is deployed on each sensor to extract feature maps from its local observation. They are subsequently uploaded to the server for aggregation and remote inference using a large-scale NN model. The aggregation of multiple observations, known as multi-view pooling, can increase the sensing accuracy as the number of views grows\cite{Chen2023arxiv}. Efforts on overcoming the communication bottleneck caused by the wireless uploading of high-dimensional features from numerous sensors have given rise to a rapidly growing research area known as communication-efficient distributed sensing\cite{Shao2023TWC,JSCC_Imp2,Deniz2023WCL,ZWUltraLoLa}. The primary goal of relevant techniques is to optimize \emph{end-to-end} (E2E) performance for a given sensing task, focusing on metrics such as E2E latency and sensing accuracy. At the link level, reliable feature transmission can be made more efficient through task-oriented compression\cite{Shao2022JSAC,Niu2019Infocom,Zhou2020IoTJ} and progressive transmission\cite{Lan2023TWC}. From a multi-access perspective, researchers have developed various task-oriented schemes that jointly consider multi-access and feature aggregation. These include distributed information bottleneck\cite{Shao2023TWC}, importance-aware unequal error protection\cite{JSCC_Imp2}, distributed JSCC\cite{Deniz2023WCL}, and short packet transmission\cite{ZWUltraLoLa}. One promising class of techniques, called \emph{over-the-air computation} (AirComp), exploits the waveform superposition property of wireless channels to achieve simultaneous access and over-the-air feature fusion, thereby addressing the scalability limitation of traditional orthogonal access\cite{GX2019IOTJ,Huang2023TWC,Wen2023TWC,Dingzhu24TWC}. The mentioned existing work typically assumes that the set of active sensors have data that are uniformly useful for the given task and that their channel states are acceptable. However, when the number of sensors is large and radio resources are limited, selecting the appropriate sensors becomes crucial to ensure the system’s efficiency and reliability.

Sensor selection is a classic technique that has been extensively studied in the area of wireless sensor networks for distributed estimation of a common source from noisy observations. Due to the limited battery life of sensors and a constraint on radio resources, sensor selection is used as a mechanism to reduce communication energy consumption, thereby prolonging the network's lifetime. The primary approach involves utilizing the well-known concept of multiuser channel diversity to select sensors with favourable channel conditions for transmission\cite{knopp1995opp,dumb_antenna}. Researchers have theoretically characterized the relationship between the expected estimation error and the number of active sensors\cite{Leong2011TIT,Leong2011TSP}. Additionally, noise levels in observations have been considered in the development of greedy sensor selection schemes\cite{Feng2013TSP}. Other factors, such as residual battery energy, have also been incorporated into sensor-selection designs with the goal of further extending the network lifetime\cite{WSNLifetime1,WSNLifetime2}. To some extent, researchers have considered primitive forms of semantic relevance in sensor selection. For instance, efforts have been made to combine channel states and data values for sensor selection in linear estimation\cite{Nordio2015TSP}.

Recent advancements in AI algorithms and data analytics provide methods to characterize and compute semantic relevance in ISEA systems, thereby enhancing their efficiency and performance. In practice, the level of relevance between sensory observations and the sensing task can vary significantly across sensors and over time due to factors such as random movements, limited sensing coverage, and occlusions. For example, in a UAV-based person tracking task, some UAVs may have their views obstructed by buildings or may capture irrelevant individuals, leading to semantically irrelevant observations\cite{UAV_ReID}. Consequently, sourcing data from all sensors can not only incur excessive overhead but also degrade the performance of downstream tasks. This challenge has motivated computer scientists to design semantic matching algorithms for selecting semantically relevant sensors via a query-and-feedback protocol, which is also adopted in this work\cite{liu2020who2com,ReID_CVPR,SEMDAS}. Specifically, given a sensing task, an edge server encodes a low-dimensional semantic query vector to the sensors as a task description; each sensor then computes and feeds back a matching score for its local observations, which is used for sensor selection to participate in the sensing task. Leveraging these scores, various simple selection schemes have been designed, such as top-1 selection\cite{liu2020who2com,ReID_CVPR}, threshold-based selection\cite{LYC2020CVPR,AndersSemRA}, or top-$k$ selection\cite{ReID_TIP}. However, existing work typically treats the communication links as reliable bit pipes without considering channel fading. In the presence of fading, it is crucial to jointly consider channel states and available radio resources along with semantic relevance to optimize the ISEA system's E2E performance and efficiency. Otherwise, despite their semantic relevance, uploaded features can be severely distorted, leading to poor sensing accuracy. In this regard, the recent work closely related to the current study is\cite{QueryRRM_TMC}, where sensors are selected using a metric that considers the difference between the semantic relevance score and transmission energy cost.

Overall, a systematic study of optimal semantic-aware sensor selection for an ISEA system is still lacking. The main challenges are twofold. First, there are few results on mathematically characterizing the effects of semantic relevance on E2E sensing performance. Deriving such results is not straightforward, as it involves the interplay of communication and AI theories. Second, it remains unknown what the optimal strategy is for balancing the considerations of semantic relevance and favorable channel states to maximize E2E sensing accuracy. This work aims to address these challenges by presenting a framework for semantic-relevance-aware sensor selection. The framework is based on a modified version of the mentioned query-and-feedback protocol where the sensor feedback includes both semantic-relevance scores and channel states, the server performs both sensor selection and radio resource allocation, and a latency constraint is enforced targeting a mission-critical application. The main contributions and findings of this work are summarized as follows.
\begin{itemize}
    \item \textbf{E2E Sensing Accuracy Analysis.} For tractability, we adopt the popular \emph{Gaussian Mixture} (GM) model for feature distribution to derive a lower bound on the E2E sensing accuracy as its tractable surrogate. This is a function of the relevance scores of selected sensors and the number of transmitted features. The derivation is non-trivial and involves two main steps. First, the accuracy conditioned on the relevance of selected sensors is lower-bounded in closed-form using the classification margin theory. Second, the semantic relevance score is related to the posterior probability of each sensor's relevance. Based on these results, the derived E2E expected sensing accuracy reveals that the contribution of each sensor to accuracy is determined by its expected classification margin, which is an increasing function of its relevance score.
    \item \textbf{Optimal Semantic-Relevance Aware Sensor Selection. } With the objective of maximizing the preceding accuracy surrogate, the problem of optimal sensor selection, which is an integer program, is solved by a relaxation leveraging the monotonicity of expected classification margin with respect to the relevance score. The optimal solution is proved to have a priority-based structure—ranking sensors according to a derived priority indicator that considers both channel state and semantic relevance score, and selecting a set of top-ranked sensors to participate in the sensing task. This leads to a low-complexity scheme for semantic-relevance-aware sensor selection.
    \item \textbf{Experiments.} Extensive experiment results are provided to validate the performance of the proposed sensor selection scheme against benchmarks on both GM-distributed synthetic datasets and real datasets (i.e., ModelNet\cite{ModelnetPaper}). The results demonstrate superior E2E accuracy of the proposed scheme as opposed to both only channel-based and semantic-based selection schemes. 
\end{itemize}

%The rest of this paper is organized as follows. Section~\ref{sec: model_metrics} introduces relevant models and performance metrics. The semantic-relevance based sensor-selection protocol is presented in Section~\ref{sec: protocol}. Then, we conduct E2E sensing accuracy analysis in Section~\ref{sec: acc_analysis}, and proceed to solve the sensor-selection problem in Section~\ref{sec: optimal_sel_algo}. Experimental results are provided in Section~\ref{sec: exp_results}, followed by the conclusion in Section~\ref{sec: conclusion}.  

\section{Models and Metrics}\label{sec: model_metrics}
Consider an ISEA system comprising a \emph{server} and $M$ \emph{sensors}, as illustrated in Fig.~\ref{fig: system}. The server receives  a query image of a target object and acquires relevant observations from sensors for accurate object classification. Due to the randomness in object locations and sensory views, only a subset of sensors capture views of the target object while others obtain irrelevant views. Thus, the server adopts the semantic-relevance based sensor selection protocol (see details in Section~\ref{sec: protocol}) to select relevant sensors. Relevant models and metrics are described in the following subsections. 

\begin{figure*}[t]
    \centering
\includegraphics[width=0.85\textwidth]{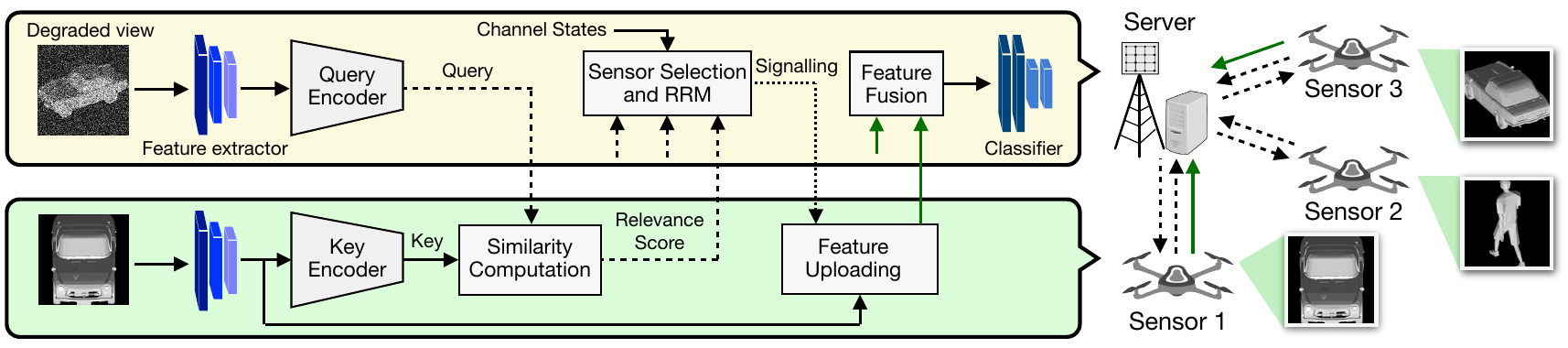}
    \caption{An ISEA system with semantic-relevance based sensor selection.}
    \label{fig: system}
\end{figure*}
\subsection{Sensing Model}
Multiple objects are present in the sensing region, each belonging to one of the $L$ object classes. The server receives in advance a query image of the target object, which is potentially degraded due to, e.g., hardware misfunctioning or occlusion. Each sensor observes one of the multiple objects depending on its current location and view angle. Let $I_m$ be an indicator of sensor relevance, where $I_m=1$ if sensor $m$ observes the target object and $I_m=0$ if it observes irrelevant objects. We assume independent observations across sensors and define $\pi_\mathsf{r}\triangleq \mathrm{Pr}(I_m=1)$ as the prior probability of each sensor observing the target object. Let $\ell_0$ denote the ground-truth class of the target object of interest to the server with a uniform prior distribution, i.e., $\mathrm{Pr}(\ell_0=\ell)=1/L$ for all $\ell = 1,\ldots,L$. It is assumed that other objects belong to classes different from $\ell_0$. Each sensor feeds its observation, typically an image of its observed object, into a feature extraction model to obtain a length-${D}$ feature vector $\mathbf{f}_m\in\mathbb{R}^D$. For tractability, we adopt the well-established GM model for multi-view sensing features\cite{Lan2023TWC,ZWUltraLoLa}. Under the GM model, each class, say class $\ell$, is represented in the feature space by a Gaussian distribution $\mathcal{N}(\bmu_\ell,\bC)$ where $\bmu_\ell\in\mathbb{R}^D$ denotes the $\ell$-th class centroid and $\bC\in\mathbb{R}^{{D}\times{D}}$ denotes a positive semi-definite covariance matrix. Without loss of generality, we assume that $\bC$ is diagonalized via \emph{principal component analysis} (PCA). Then, conditioned on the ground-truth class of the target object, $\ell_0$, the features obtained by sensor $m$ follows a GM distribution as
\begin{equation}\label{eq:feature_distribution}
\begin{split}
    &p(\bff_m|\ell_0) \\&= \mathrm{Pr}(I_m{=}\,1)p(\bff_m|I_m{=}\,1,\ell_0) + \mathrm{Pr}(I_m{=}\,0)p(\bff_m|I_m{=}\,0,\ell_0),
    \\&={\pi_\mathsf{r} \underbrace{\mathcal{N}(\bff_m|\bmu_{\ell_0},\bC)}_{\substack{\text{distribution of}\\ \text{relevant views}}}} + {{(1-\pi_\mathsf{r})} \underbrace{\sum_{\ell\neq\ell_0}\frac{1}{L-1}\mathcal{N}(\bff_m|\bmu_{\ell},\bC)}_{\text{distribution of irrelevant views}}}. 
\end{split}
\end{equation}
\subsection{Semantic Matching Model}\label{subsec: sem_match_model}
Intuitively, only features from sensors with $I_m=1$ contribute to recognizing the target object, while features from other sensors are redundant or even potentially harmful if incorporated into feature fusion. However, the ground-truth sensor relevance $I_m$ is not known to the server or sensors. An existing \emph{semantic matching} scheme is adopted to identify sensors with relevant observations\cite{LYC2020CVPR}. The procedure is described as follows. The server utilizes a query encoder, $G_\mathsf{q}(\cdot)$, to encode features of the query image, $\bff_0\in\mathbb{R}^D$, into a $D_\mathsf{q}$-dimensional semantic query vector. This resulting query vector is denoted as $\bq = G_\mathsf{q}(\bff_0)$. Similarly, each sensor uses a key encoder, $G_\mathsf{k}(\cdot)$, to encode its observation into a $D_\mathsf{q}$-dimensional key, denoted by $\bk = G_\mathsf{k}(\bff_m)$. Under the GM model, both the query and key encoders are defined as linear transformations, i.e., $\bq = \bW_\mathsf{q}\bff_0$ and $\bk_m = \bW_\mathsf{k}\bff_m$, where $\bW_\mathsf{q}\in\mathbb{R}^{D_\mathsf{q}\times D}$ and $\bW_\mathsf{k}\in\mathbb{R}^{D_\mathsf{q}\times D}$. The query is broadcast to all sensors for local semantic matching, wherein a semantic relevance score is computed as a function of the received query and the locally computed key. We use the widely adopted \emph{dot-product attention} to compute the semantic relevance score for each sensor $m$ as $\phi_m = \bq^T \bk_m$\cite{liu2020who2com,LYC2020CVPR,ReID_CVPR}.
It can be seen that semantic matching projects the query image and sensor observations into a semantic space, where a cosine-similarity score is computed between each query-key pair. The score $\phi_m$ indicates the relevance level of sensor $m$: a larger $\phi_m$ suggests a higher probability that sensor $m$ is relevant, i.e., $I_m=1$, and vice versa. All scores $\{\phi_m\}_{m=1}^M$ are fed back to the server for scheduling and resource allocation decisions. 

\subsection{Feature Pruning-and-Fusion Model}\label{subsec: feature_prune_fuse}
Features from multiple sensors are uploaded to the server and fused into a global feature map such that multi-view diversity is exploited to improve the classification accuracy\cite{Chen2023arxiv}.  Feature pruning is employed to ensure that feature uploading completes within the given latency constraint. Let $\cS\in\{1,\ldots,M\}$ represent the subset of sensors selected by the server for feature uploading. Under the server's coordination, each sensor in $\cS$ selects a subset of feature dimensions for uploading, denoted by $\tilde{\cD}\subseteq \{1,\ldots, D\}$. Given the number of selected features $|\tilde{\cD}|=\tilde{D}$, two feature ordering schemes can be utilized to specify $\tilde{\cD}$, as detailed below.
\begin{itemize}
    \item  {\bf Random ordering:} The server randomly selects $\tilde{D}$ dimensions from $\{1,\ldots,D\}$, resulting in the feature subset $\tilde{\cD}=\tilde{\cD}_{\sf rnd}$. This scheme is suitable when a lookup table for dimension-wise feature importance is not available due to the profiling overhead, or when features have approximately equal importance due to, e.g., feature whitening\cite{FeatWhitening}.
    \item  {\bf Importance ordering:} The server selects $\tilde{D}$ dimensions from $\{1,\ldots,D\}$ with top-$\tilde{D}$ feature importance level. This scheme requires the server to maintain a lookup table indicating the importance level of each feature dimension. For GM-distributed data, the importance level of the $d$-th dimension, denoted as $\bar{g}(d)$, is mathematically defined as
    \begin{equation}
        \bar{g}(d) = \frac{2}{L(L-1)}\sum_{\ell<\ell'\leq L}\frac{(\mu_{\ell,d}-\mu_{\ell',d})^2}{C_{d,d}},
    \end{equation}
    which is known as the average \emph{discriminant gain} (DG)\cite{Lan2023TWC}. It characterizes the class separability on dimension $d$, averaged over all class pairs. 
    With importance ordering, the resultant feature subset, $\tilde{\cD}=\tilde{\cD}_{\sf imp}$, consists of $\tilde{D}$ feature dimensions with the highest $\bar{g}(d)$ values.
    
\end{itemize}
Let $\tilde{\bff}_m$ denote the subvector of $\bff_m$ in a reduced feature space designated by $\tilde{\cD}$. Mathematically, define a down-sampling matrix $\tilde{\mathbf{D}}=[\bee_{r_1},\bee_{r_2},\ldots,\bee_{r_{\tilde{D}}}]^T\in \{0,1\}^{\tilde{D}\times D}$  where $\bee_{n}$ is a length-$D$ unit vector with $1$ in its $n$-th dimension, $r_i\in\tilde{\cD}$, and $r_i\leq r_j$ for all $i<j$. We then have $\tilde{\bff}_m = \tilde{\mathbf{D}}{\bff}_m$. The \emph{attentive fusion} scheme performs a weighted sum over the pruned local features, such that $\tilde{\bff}_m$ is assigned a larger weight if the relevance score $\phi_m$ is larger \cite{LYC2020CVPR}. The attentively fused feature is given by
\begin{equation}\label{eqn: attn_fusion}
    \tilde{\bff}_\mathsf{g} = \sum_{m\in\cS}{\omega}_m\tilde{\bff}_m,
\end{equation}
where ${\omega}_m\triangleq\frac{\exp(\phi_m/\tau)}{\sum_{m\in\cS}\exp(\phi_m/\tau)}$ is the \emph{softmax} attention weight and $\tau$ is a pre-determined temperature parameter controlling the smooth transition of $\{{w}_m\}_{m\in\cS}$ from binary to uniform
distributions. It is worth pointing out that average view-pooling is a special case of attentive fusion as $\tau\to \infty$.

\subsection{Communication Model}
The communication operations in the query-based ISEA system consist of downlink query broadcasting, semantic relevance score feedback, and feature uploading. Since the query is a low-dimensional vector and the relevance scores are scalars, the communication overhead of the first two operations is considered negligible. Consider a \emph{time-division multiple access} (TDMA) system for feature uploading\footnote{We acknowledge that AirComp can be utilized to achieve simultaneous feature uploading, thereby reducing multi-access latency\cite{Huang2023TWC}. However, in our context, the number of selected sensors is usually limited (e.g., one to three), resulting in a relatively small benefit from AirComp when compared to the synchronization overhead and hardware compatibility issues it presents.}. Time is divided into slots, each spanning a duration of $T$ seconds, which is chosen to meet the latency requirement of the ISEA task. Within each slot, time is allocated to the selected sensors, and the allocated time for sensor $m$ is denoted by $t_m$. We assume block fading such that the channel gain for sensor $m$, denoted as $h_m$, is constant during the entire slot. Define $P_m$ as the transmission power for sensor $m$. Then, its data rate for feature uploading, denoted as $r_m$, is given as
\begin{equation}
    r_m = B\log_2\left(1+\frac{P_m |h_m|^2}{N_0}\right),
\end{equation}
where $B$ and $N_0$ denote the uplink bandwidth and noise power, respectively. Each element of the feature vector is quantized in to $Q$ bits, which is sufficiently high such that quantization noise is negligible. Given the number of selected feature dimensions, $\tilde{D}$, the time allocated for sensor $m$ to upload its features is then given by $t_m = \frac{Q\tilde{D}}{r_m}$. The total communication time constraint is then given by
\begin{equation}\label{eq:comm_constraint}
    \sum_{m\in\mathcal{S}}\frac{Q\tilde{D}}{r_m}\leq T.
\end{equation}

\subsection{Inference Model and Metrics}
The server inputs the fused global feature vector $\tilde{\bff}_\mathsf{g}$ into a downstream classification model to infer the class of the target object. Two types of classifiers are considered and described as follows.
\begin{itemize}
    \item  {\bf Linear classifier:} For GM-distributed data in \eqref{eq:feature_distribution}, a linear \emph{maximum likelihood} (ML) classifier is employed to infer the object class, which is equivalent to the \emph{maximum a posteriori} (MAP) classifier due to uniform prior of the ground-truth class. Due to feature selection, the classification is based on decision boundaries in the reduced feature space defined by $\tilde{\cD}$. Let $\tilde{\bmu}_\ell\triangleq \tilde{\mathbf{D}}{\bmu}_\ell$ for all $\ell$ and $\tilde{\bC}\triangleq \tilde{\mathbf{D}}\bC$ denote the class centroids and data covariance in the reduced feature space, respectively. Then, the inferred class, denoted as $\hat{\ell}$, is given by
    \begin{equation}\label{eq:linear_classifier}
        \hat{\ell}=\mathop{\arg\max}_{\ell} \log p(\tilde{\bff}_\mathsf{g}|{\ell})=\mathop{\arg\min}_{\ell} z_{\ell} (\tilde{\bff}_\mathsf{g}),
    \end{equation}
    where $z_{\ell} (\tilde{\bff}_\mathsf{g})\triangleq(\tilde{\bff}_\mathsf{g}-\tilde{\bmu}_{\ell})^T \tilde{\mathbf{C}}^{-1}(\tilde{\bff}_\mathsf{g}-\tilde{\bmu}_{\ell})$ is the squared \emph{Mahalanobis distance} between $\tilde{\bff}_\mathsf{g}$ and $\tilde{\bmu}_{\ell}$ in the reduced feature space.
    \item  {\bf NN classifier:} For real-world data, $\tilde{\bff}_\mathsf{g}$ is input into an NN classifier consisting of, e.g., convolutional layers and feed-forward layers, to output the predicted class $\hat{\ell}$.
\end{itemize}

The E2E performance metric for the downstream classification task is \emph{expected sensing accuracy}, defined as the probability of correct classification averaged over random observations of sensors. This randomness arises from two sources, i.e., the uncertainty regarding the relevance of selected sensors, and the data noise in the feature space. The former can be characterized by the posterior probability of sensor relevance, which is related to the relevance score $\phi_m$ and denoted by $\pi_m\triangleq \mathrm{Pr}(I_m=1|\phi_m)$. The expected sensing accuracy as a function of $\cS$ and $\tilde{\cD}$, denoted by $A(\cS,\tilde{\cD}|\bphi)$, is then expressed as
\begin{align}\label{eq:exp_acc}
    A(\cS,\tilde{\cD}|\bphi)&=\frac{1}{L}\sum_{\ell_0=1}^L \mathrm{Pr}(\hat{\ell}=\ell_0|\ell_0,\bphi)\nonumber\\
    & = \frac{1}{L}\sum_{\ell_0=1}^L \mathbb{E}[\mathrm{Pr}(\hat{\ell}=\ell_0|\ell_0,\{I_m\})|\bphi],
\end{align}
where $\bphi\triangleq[\phi_1,\ldots,\phi_M]^T$.
Given one of the feature ordering schemes outlined in Section~\ref{subsec: feature_prune_fuse}, $\tilde{\cD}$ is fully determined once the number of selected features $\tilde{D}$ is specified. With a slight abuse of notations, the expected sensing accuracy can be expressed as a function of $\mathcal{S}$ and $\tilde{D}$:
\begin{equation}\label{eqn: def_exp_acc}
    A(\cS,\tilde{D}|\bphi)=\begin{cases}
        A(\cS,\tilde{\cD}_{\sf rnd}|\bphi), & \text{random ordering,}\\
        A(\cS,\tilde{\cD}_{\sf imp}|\bphi), & \text{importance ordering.}
    \end{cases}
\end{equation}

\begin{figure}[t]
    \centering
\includegraphics[width=0.28\textwidth]{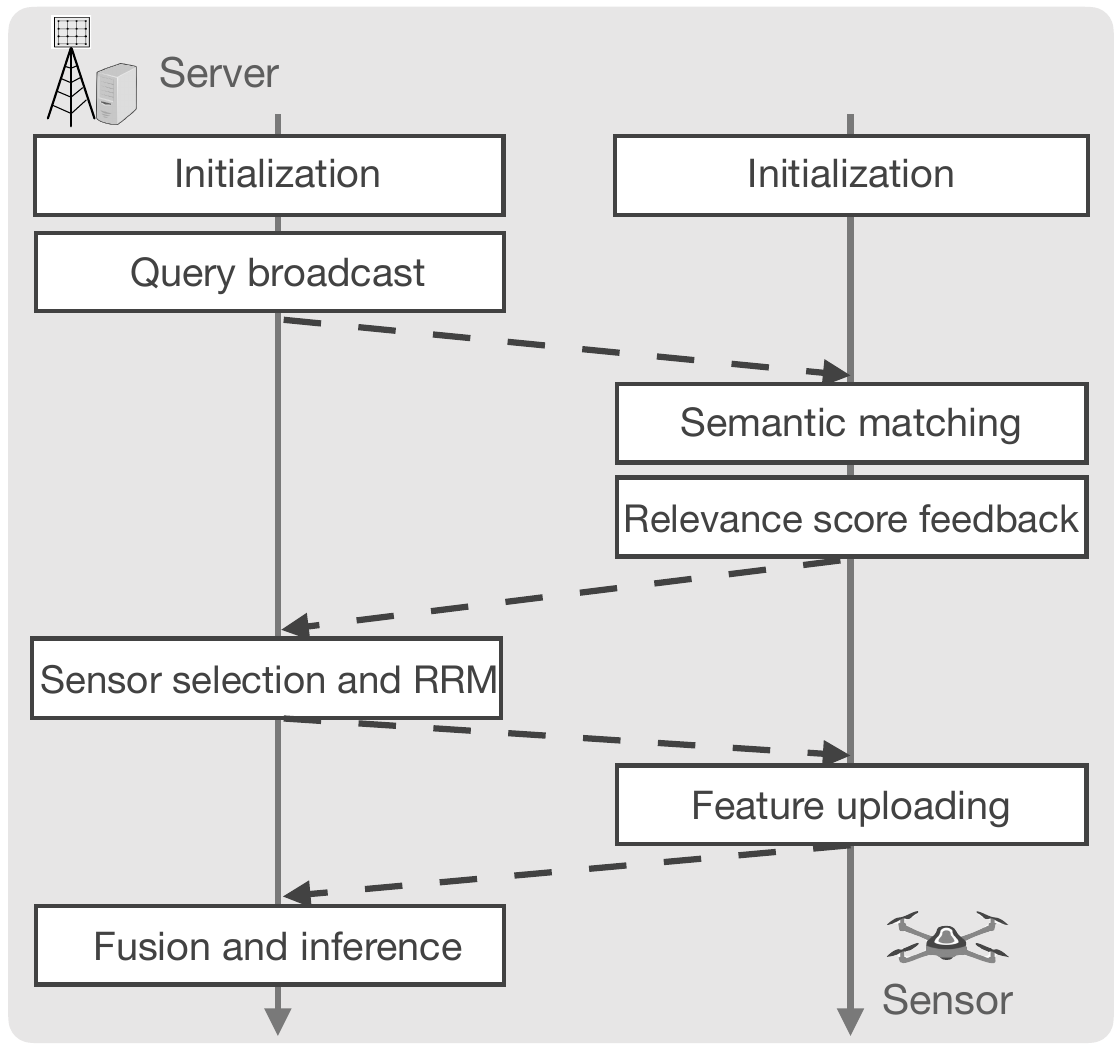}
    \caption{The proposed semantic-relevance based sensor selection protocol.}
    \label{fig: protocol}
\end{figure}

\section{Semantic-Relevance Based Sensor Selection Protocol}\label{sec: protocol}
The semantic-relevance based sensor selection protocol is executed for each sensing instance initiated by the server. The overview of the proposed protocol is illustrated in Fig.~\ref{fig: protocol}, and its detailed steps are listed sequentially below.
\subsubsection{Query broadcast} The server acquires a query image, encoding it into a semantic query vector $\bq=G_\mathsf{q}(\bff_0)$. The query is then broadcast to all $M$ sensors.
\subsubsection{Semantic matching and feedback}
Each sensor, say sensor $m$, captures a local observation $\bff_m$ with unknown relevance to the server's query. Upon receiving $\bq$, sensor $m$ encodes its observation into a key $\bk_m=G_\mathsf{k}(\bff_m)$, and computes the relevance score $\phi_m=\bq^T \bk_m$. Then, sensor $m$ feeds back the relevance score $\phi_m$ and its channel state $h_m$ to the server.
\subsubsection{Sensor selection and RRM} Based on the received relevance score $\{\phi_m\}$ and channel states $\{h_m\}$, the server determines a subset of sensors $\cS$,  the selected feature subset $\tilde{\cD}$, and the allocated uploading time $t_m$ for each $m\in \cS$. The decisions are sent to each selected sensor.
\subsubsection{Feature uploading} Each sensor in $\mathcal{S}$ prunes its local observation $\bff_m$, obtaining the pruned feature $\tilde{\bff}_m=\tilde{\bD}\bff_m$. Subsequently, it uploads $\tilde{\bff}_m$ within the timeslot allocated by the server. The total uploading time shall not exceed the latency constraint $T$, as otherwise the latency-sensitive task shall be considered as failed.
\subsubsection{Feature fusion and inference}
The server receives the uploaded features $\{\tilde{\bff}_m\}$, and performs attention fusion as described in \eqref{eqn: attn_fusion} to obtain the fused global feature $\tilde{\bff}_\mathsf{g}$. This fused global feature is then input into the downstream classifier to output the inferred class $\hat{\ell}$.

\section{E2E Sensing Accuracy Analysis}\label{sec: acc_analysis}
In this section, we analyze the expected sensing accuracy, $A(\cS,\tilde{D}|\bphi)$, in \eqref{eqn: def_exp_acc} achieved by selecting a subset of sensors, $\cS$, and number of features, $\tilde{D}$, for the linear classifier in \eqref{eq:linear_classifier}. The exposition hinges on two key results presented in Section~\ref{subsec: conditional_acc} and Section~\ref{subsec: prob_relevance}, respectively. First, a lower bound is first provided for the sensing accuracy conditioned on sensor relevance, i.e., $\{I_m\}$. Then, explicit expressions are derived for the posterior probability of sensor relevance, $\mathrm{Pr}(I_m=1|\phi_m)$. Finally, we combine the two results to derive our main result on $A(\cS,\tilde{D}|\bphi)$ for both random and importance feature ordering in Section~\ref{subsec: exp_sensing_acc}.

\subsection{Sensing Accuracy Conditioned on Sensor Relevance}\label{subsec: conditional_acc}
We first consider the sensing accuracy conditioned on sensor relevance $\{I_m\}$, which is a function of selected sensors $\cS$ and feature dimensions $\tilde{\cD}$. The conditional accuracy is defined as
\begin{equation}
     A(\cS,\tilde{\cD}|\{I_m\}) = \frac{1}{L}\sum_{\ell_0=1}^L \mathrm{Pr}(\hat{\ell}=\ell_0|\ell_0,\{I_m\}).
\end{equation}
To analyze this accuracy, we start by examining the distribution of the aggregated feature $\tilde{\bff}_\mathsf{g}=\sum_{m\in\cS}{\omega}_m\tilde{\bff}_m$ conditioned on the ground-truth class $\ell_0$ and sensor relevance $\{I_m\}$. Consider the component from sensor $m$ in $\tilde{\bff}_\mathsf{g}$. When $I_m=1$, the feature uploaded by sensor $m$, $\tilde{\bff}_m$, follows a Gaussian distribution $\tilde{\bff}_m\sim\mathcal{N}(\tilde{\bmu}_{\ell_0},\tilde{\mathbf{C}})$, where $\tilde{\bmu}_{\ell_0}$ is the ground-truth class centroid. On the other hand, when $I_m=0$, $\tilde{\bff}_m$ is distributed as $\tilde{\bff}_m\sim\mathcal{N}(\tilde{\bmu}_{\ell_m},\tilde{\mathbf{C}})$, with $\tilde{\bmu}_{\ell_m}$ being the centroid of sensor $m$'s observed class, different from ${\ell_0}$.
Conditioned on the observed class of irrelevant sensors $\{\ell_m\}_{I_m=0}$, the aggregated feature $\tilde{\bff}_\mathsf{g}$ is the sum of $|\cS|$ Gaussian random variables, which is also Gaussian distributed as
\begin{equation}\label{eq:cond_agg_feature_dist}
\begin{split}
    p&(\tilde{\bff}_\mathsf{g}|\ell_0,\{I_m\},\{\ell_m\})\\& = \mathcal{N}\Bigg(\tilde{\bff}_\mathsf{g}\Bigg|\underbrace{\rho\tilde{\bmu}_{\ell_{0}}}_{\substack{\text{relevant}\\ \text{sensors}}}+\underbrace{\sum_{m\in\mathcal{S}}(1-I_m)w_m\tilde{\bmu}_{\ell_{m}}}_{\text{irrelevant sensors}},\eta\tilde{\mathbf{C}}\Bigg.\Bigg),
\end{split}
\end{equation}
where $\rho\triangleq \sum_{m\in\mathcal{S}}I_m{w}_{m}\in [0,1]$ is the proportion of ground-truth feature mean in the aggregated feature mean, and $\eta\triangleq\sum_{m\in\mathcal{S}}{w}_{m}^2\in \left[{|\cS|}^{-1},1\right]$ is a denoising factor characterizing data-noise reduction from multi-view aggregation. When $\cS$ includes irrelevant sensors, i.e., $\{m|m\in\cS,I_m=0\}\neq\emptyset$, the aggregated feature mean deviates from the ground-truth mean $\tilde{\bmu}_{\ell_0}$, which intuitively degrades the sensing accuracy. This effect is quantified in the following theorem, which establishes a lower bound on the conditional sensing accuracy. 
\begin{theorem}
    \emph{\emph{(Lower Bound on Conditional Accuracy)} Given selected sensors, $\cS$ and selected features, $\tilde{\cD}$, the classification accuracy conditioned on the sensor relevance, $A(\cS,\tilde{\cD}|\{I_m\})$, is lower bounded by }
     \begin{align}\label{eq: lb_cond_acc}
        A&(\cS,\tilde{\cD}|\{I_m\})\nonumber\\&\geq 1-(L-1)Q\left[\frac{1}{\sqrt{\eta}}\biggl(\frac{\sqrt{\tilde{G}_{\min}}}{2}-2(1-\rho)\tilde{\delta}_{\max} \biggr)\right],\nonumber\\
        &\triangleq A_{\sf lb}(\cS,\tilde{\cD}|\rho),
    \end{align}
    \emph{where $\tilde{G}_{\min}\triangleq\min_{\ell\neq\ell'}\Vert\tilde{\bmu}_{\ell}-\tilde{\bmu}_{\ell'}\Vert^2_{\tilde{\mathbf{C}}}=\min_{\ell\neq\ell'}\sum_{d\in\tilde{\mathcal{D}}}\frac{(\mu_{\ell,d}-\mu_{\ell',d})^2}{C_{d,d}}$ is the minimum Mahalanobis distance between centroids of all class pairs, also known as the minimum pairwise DG\cite{Lan2023TWC}; $\tilde{\delta}_{\max}\triangleq \max_{\ell}\Vert\tilde{\mu}_{\ell}\Vert_{\tilde{\mathbf{C}}}=\max_{\ell}\sqrt{\sum_{d\in\hat{\mathcal{D}}}\frac{\mu^2_{\ell,d}}{C_{d,d}}}$ is the maximum of Mahalanobis norms of all class centroids.}
\end{theorem}

\begin{proof}
    (See Appendix A.)
\end{proof}
\begin{figure}[t]
    \centering
\includegraphics[width=0.485\textwidth]{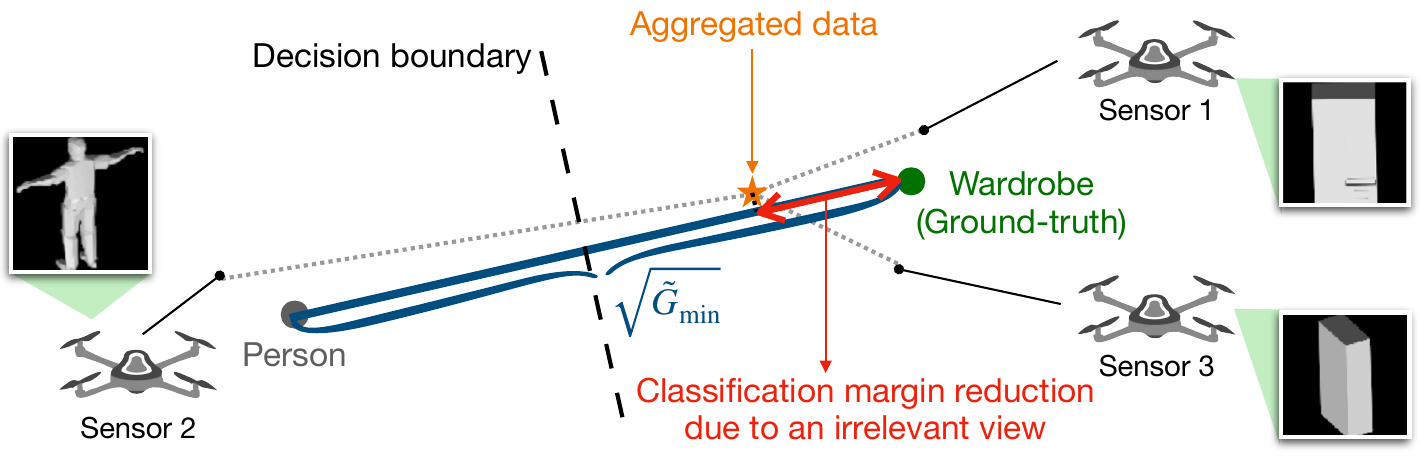}
    \caption{Illustration of classification margin reduction incurred by irrelevant views in a binary-class case. Therein, the red arrow reflects the margin reduction, which is quantified by $2(1-\rho)\tilde{\delta}_{\max}$ in Theorem 1.}
    \label{fig: margin_reduction}
\end{figure}
Theorem~1 can be interpreted as follows. Inside the Q function, $\frac{\sqrt{\tilde{G}_{\min}}}{2}$ is the minimum Mahalanobis distance from a class centroid to the decision boundary, indicating the inherent classification margin determined by separability of classes in the feature space. Sensor relevance influences the accuracy through a penalty term, $2(1-\rho)\tilde{\delta}_{\max}$, which is subtracted from the inherent classification margin, $\frac{\sqrt{\tilde{G}_{\min}}}{2}$. This term is essentially the maximum distance by which the irrelevant components can pull the aggregated feature mean towards the decision boundary, as illustrated by a red arrow in Fig.~\ref{fig: margin_reduction}. This reduction in margin decreases the probability of correct classification. When all sensors are relevant, we have $\rho=1$, and thus the penalty term becomes zero, preserving the inherent classification margin. On the other hand, the denoising factor $\eta\leq 1$ quantifies the variance reduction achieved through  multi-view aggregation. Incorporating a larger number of relevant sensors in aggregation can result in a smaller $\eta$, indicating that the aggregated feature is more concentrated around its mean. This improves the accuracy by reducing the probability of crossing the decision boundary.

\subsection{Posterior Probability of Sensor Relevance}\label{subsec: prob_relevance}
To derive the expected sensing accuracy, $A(\cS,\tilde{\cD}|\bphi)$, from the conditional accuracy, it is necessary to characterize the posterior probability of sensor relevance, i.e., $\pi_m \triangleq \mathrm{Pr}(I_m=1|\phi_m)$. Based on the semantic matching model in Section~\ref{subsec: sem_match_model}, the exact expression of posterior probability for GM-distributed data is given by the following lemma. 
\begin{lemma}
\emph{The posterior probability of sensor relevance, $\pi_m= \mathrm{Pr}(I_m=1|\phi_m)$, is given by}
    \begin{equation}\label{eqn: post_prob}
    \pi_m =\frac{1}{1+\frac{1-\pi_\mathsf{r}}{\pi_\mathsf{r}(L-1)}\sum_{\ell\neq \ell_0}\exp\left[-\frac{\alpha_{\ell_0,\ell}(\phi_m-\bar{\phi}_{\ell_0,\ell})}{\sigma_\mathsf{s}^2}\right]},
\end{equation}
\emph{where $\alpha_{\ell_0,\ell}\triangleq\bq^T\bW_\mathsf{k}(\bmu_{\ell_0}-\bmu_{\ell})$ is a decay factor as the average difference in relevance scores between relevant and irrelevant cases, $\bar{\phi}_{\ell_0,\ell} \triangleq \frac{\bq^T\bW_\mathsf{k}(\bmu_{\ell_0}+\bmu_{\ell})}{2}$ is the mid-point in relevance scores between the two cases, and $\sigma_\mathsf{s}^2=\bq^T\bW_{\mathsf{k}}^T\bC\bW_{\mathsf{k}}\bq$ is the relevance-score variance.}
\end{lemma}
\begin{proof}
    (See Appendix B.)
\end{proof}

Three key properties can be observed from the expression of $\pi_m$. First, when $\alpha_{\ell_0,\ell}\geq 0$ for all $\ell\neq \ell_0$, $\pi_m$ is an increasing function in $\phi_m$, indicating that a higher semantic relevance score leads to a higher probability of semantic relevance. The condition, $\alpha_{\ell_0,\ell}\geq 0$ for all $\ell\neq \ell_0$, can be interpreted as \emph{query effectiveness}, as it requires the query to be more closely aligned with the relevant data centroid than with the irrelevant centroid in the semantic matching space. Second, as the query aligns more closely with relevant samples in the semantic matching space, the values of ${\alpha_{\ell_0,\ell}}$ increase, causing $\pi_m$ to transition more rapidly from 0 to 1 as the relevance score increases. Third, the value of $\pi_m$ is invariant to linear scaling of either the query or key encoders, as both the numerator and denominator in the exponential functions are quadratic in both $\bW_\mathsf{q}$ and $\bW_\mathsf{k}$. This conforms with the intuition that applying a linear scaling factor to all semantic relevance scores should not change the implied probability of relevance.

However, the server cannot directly evaluate $\pi_m$  due to lack of knowledge of $\ell_0$. Thus we propose the following estimation based on mean approximation. Specifically, the decay factor $\alpha_{\ell_0,\ell}$ is approximated by $\bar{\alpha}$, which is given by
\begin{equation}
    \bar{\alpha}=\mathbb{E}_{\bq,{\ell_0}}\biggl[\frac{1}{L-1}\sum_{\ell\neq\ell_0}\bq^T\bW_\mathsf{k}(\bmu_{\ell_0}-\bmu_{\ell})\biggr].
\end{equation}
It can be seen that $\bar{\alpha}$ is the global average relevance-score difference between relevant and non-relevant cases over all possibilities of $\ell_0$, which can be estimated empirically from training data. Similarly, $\phi_{\ell_0,\ell}$ can be approximated with $\bar{\phi}$, which is given by
\begin{equation}
    \bar{\phi} = \mathbb{E}_{\bq,{\ell_0}}\biggl[\frac{1}{2(L-1)}\sum_{\ell\neq\ell_0}\bq^T\bW_\mathsf{k}(\bmu_{\ell_0}+\bmu_{\ell})\biggr].
\end{equation}
Last, the relevance-score $\sigma_{\mathsf{s}}^2$ is approximated by $\bar{\sigma}_{\mathsf{s}}^2=\mathbb{E}_\bq[\bq^T\bW_{\mathsf{k}}\bC\bW_{\mathsf{k}}\bq]$ over the query distribution. Substituting the approximations into \eqref{eqn: post_prob} yields the final estimation of $\pi_m$, denoted as $\hat{\pi}_m$, which is a scaled \emph{sigmoid} function of the received relevance score $\phi_m$ given by
\begin{equation}\label{eqn: approx_post_prob}
    \pi_m \approx \hat{\pi}_m\triangleq\frac{1}{1+\frac{1-\pi_\mathsf{r}}{\pi_\mathsf{r}}\exp\left[-\frac{\bar{\alpha}(\phi_m-\bar{\phi})}{\bar{\sigma}_\mathsf{s}^2}\right]}.
\end{equation}
% Despite a simpler form that only requires the statistics of relevance scores, it can be verified that $\hat{\pi}_m$ is still increasing in $\phi_m$ as $\bar{\alpha}>0$ and invariant to encoder scaling.

\begin{figure}[t]
    \centering
\includegraphics[width=0.28\textwidth]{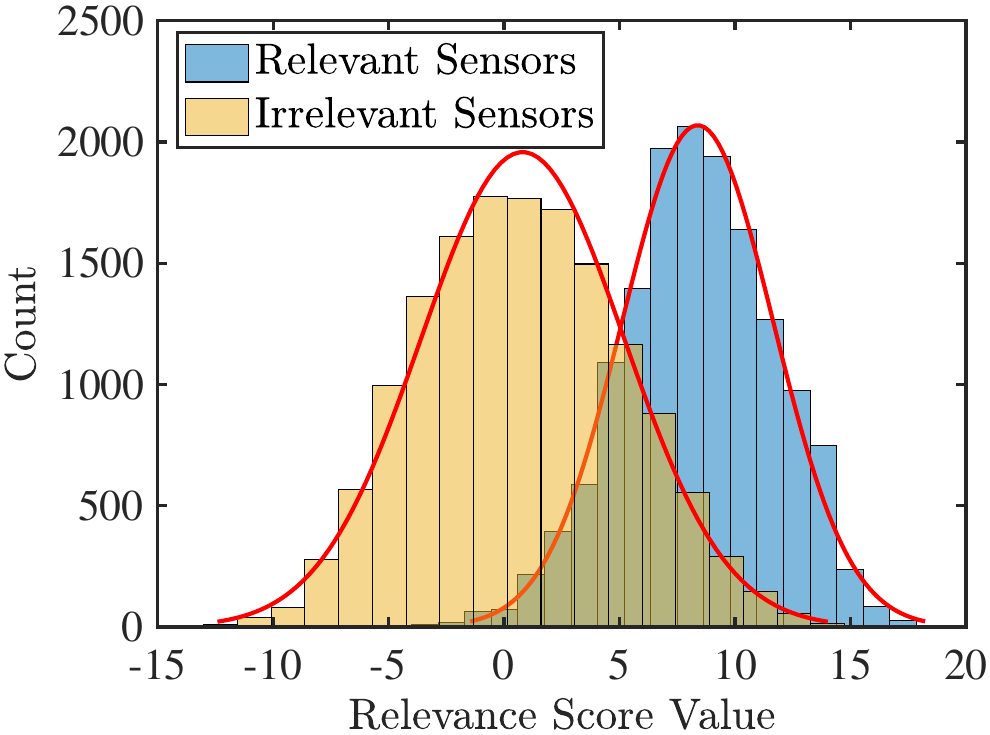}
    \caption{Distributions of relevance scores on the ModelNet training dataset fitted by Gaussian probability density functions.}
    \label{fig: matching_stat}
\end{figure}

\begin{remark}
    (Precision of Probability Estimation) \emph{The approximation in \eqref{eqn: approx_post_prob} is exact if the scores of relevant and irrelevant cases are respectively Gaussian distributed as $\mathcal{N}(\bar{\phi}_\mathsf{r},\bar{\sigma}_\mathsf{s}^2)$ and $\mathcal{N}(\bar{\phi}_\mathsf{ir},\bar{\sigma}_\mathsf{s}^2)$, with $\bar{\alpha}=\bar{\phi}_\mathsf{r}-\bar{\phi}_\mathsf{ir}>0$. As illustrated in Fig.~\ref{fig: matching_stat}, using semantic encoders trained for real dataset, i.e., ModelNet \cite{LYC2020CVPR}, the distribution of relevance scores is well approximated as conditional Gaussian. The estimated mean and variance of relevance scores on training sets can be utilized to compute \eqref{eqn: approx_post_prob} in the deployment stage. This validates the precision of the posterior probability estimation. }
\end{remark}

\subsection{Expected Sensing Accuracy}\label{subsec: exp_sensing_acc}
Next, we derive the expected sensing accuracy with results from the preceding two subsections. Using the total expectation formula, the expected sensing accuracy is expressed as
\begin{align}
    A(\cS,\tilde{\cD}|\bphi) &= \mathbb{E}_{\{I_m\}|\bphi}[A(\cS,\tilde{\cD}|\{I_m\})|\bphi]\nonumber\\
    &\geq \mathbb{E}_{\rho|\bphi}\left[\left.A_{\sf lb}(\cS,\tilde{D}|\rho)\right|\bphi\right]\\&\triangleq A_{\mathsf{lb}}(\cS,\tilde{\cD}|\bphi).\nonumber
\end{align}
Recall that the random proportion of relevant sensors, $\rho$, is defined as $\rho 
    =\sum_{m\in\mathcal{S}} I_m{w}_{m}$.
Conditioned on $\bphi$, each $I_m$ is a Bernoulli random variable with parameter $\pi_m$, and thus $\rho$ is a sum of $|\cS|$ Bernoulli independent variables. Its conditional expectation is therefore given by
\begin{equation}
    \mathbb{E}[\rho|\bphi] = \sum_{m\in\mathcal{S}} {w}_{m}\mathbb{E}[I_m|\phi_m]=\sum_{m\in\mathcal{S}} {w}_{m}\pi_m.
\end{equation}
For tractability, we adopt an approximation $ A_{\mathsf{lb}}(\cS,\tilde{\cD}|\bphi)\approx \hat{A}_{\mathsf{lb}}(\cS,\tilde{\cD}|\bphi)$ by first-order Taylor expansion around the mean of $\rho$, which can be expressed as
\begin{equation}\label{eqn: approx_exp_acc}
\begin{split}
    \hat{A}&_{\mathsf{lb}}(\cS,\tilde{\cD}|\bphi)=
    %1-(L-1)Q\left[\frac{1}{\sqrt{\eta}}\left(\frac{\sqrt{\tilde{G}_{\min}}}{2}-2(1-\mathbb{E}[\rho|\bphi])\tilde{\delta}_{\max} \right)\right]\nonumber\\
    1-(L-1)\cdot\\&Q\left[\frac{1}{\sqrt{\eta}}\biggl(\frac{\sqrt{\tilde{G}_{\min}}}{2}-2\biggl(1-\sum_{m\in\mathcal{S}} {w}_{m}\hat{\pi}_m\biggr)\tilde{\delta}_{\max} \biggr)\right].
    \end{split}
\end{equation}

The above closed-form approximation still depends on the specific subset of selected features, $\tilde{\mathcal{D}}$. We now proceed to establish this approximation as a function of the number of selected features, $\tilde{D}$, considering random and importance feature ordering. First, consider the random-ordering case, where $\tilde{\mathcal{D}}=\tilde{\cD}_{\sf rnd}$ is a randomly selected size-$\tilde{D}$ subset of $\{1,\ldots,D\}$. The selection of $\tilde{\cD}_{\sf rnd}$ affects $\hat{A}_{\mathsf{lb}}(\cS,\tilde{\cD}_{\sf rnd}|\bphi)$ through $\tilde{G}_{\min}$ and $\tilde{\delta}_{\max}$, i.e., the minimum pairwise DG and maximum feature norm on the randomly selected feature dimensions. Given that the number of features is typically large in practice (e.g., $>1000$), we derive the asymptotic behavior of $\tilde{G}_{\min}$ and $\tilde{\delta}_{\max}$ as the number of features grows, which is presented in the following lemma. The proof, based on the well-known \emph{strong law of large numbers} and the \emph{continuous mapping theorem}\cite{van2000asymptotic}, is omitted here for brevity.
\begin{lemma}
    \emph{Let $\beta=\frac{\tilde{D}}{D}$ denote a constant feature pruning ratio. Assume respective i.i.d. priors for $\mu_{\ell,d}$ and $C_{d,d}$. Then,}
    \begin{align}
        \frac{\tilde{G}_{\min}}{{G}_{\min}}\xrightarrow[D\to \infty]{a.s.}\beta,\quad \frac{\tilde{\delta}_{\max}}{{\delta}_{\max}}\xrightarrow[D\to \infty]{a.s.}\sqrt{\beta},
    \end{align}
    \emph{where ${G}_{\min}=\min_{\ell\neq\ell'}\sum_{d=1}^D\frac{(\mu_{\ell,d}-\mu_{\ell',d})^2}{C_{d,d}}$ and  ${\delta}_{\max}= \max_{\ell}\Vert\tilde{\mu}_{\ell}\Vert_{\tilde{\mathbf{C}}}=\max_{\ell}\sqrt{\sum_{d=1}^D\frac{\mu^2_{\ell,d}}{C_{d,d}}}$ represent the minimum pairwise DG and maximum Mahalanobis norm on full feature dimensions, respectively.}
\end{lemma}
We therefore utilize the approximations $\tilde{G}_{\min}\approx\beta{G}_{\min}$ and $\tilde{\delta}_{\max}\approx\sqrt{\beta}{\delta}_{\max}$, leading to
\begin{equation}
\begin{split}
    \hat{A}&_{\mathsf{lb}}(\cS,\tilde{D}|\bphi)\approx 1-(L-1)\cdot\\&Q\left[\sqrt{\frac{\tilde{D}}{\eta D}}\biggl(\frac{\sqrt{{G}_{\min}}}{2}-2\biggl(1-\sum_{m\in\mathcal{S}} {w}_{m}\hat{\pi}_m\biggr){\delta}_{\max} \biggr)\right].
\end{split}
\end{equation}
This is a tight approximation which only relies on the pruning ratio but not the deterministic DG and feature norm over selected dimensions.

Next, consider the sensing accuracy for importance ordering given $\tilde{D}$.
It depends on the minimum pairwise DG on the subset feature dimensions with top-$\tilde{D}$ importance, defined as ${G}_{\min}(\tilde{D})\triangleq\min_{\ell\neq\ell'}\sum_{d\in\tilde{\mathcal{D}}_{\sf imp}}\frac{(\mu_{\ell,d}-\mu_{\ell',d})^2}{C_{d,d}}$, as well as the maximum feature Mahalanobis norm, defined as ${\delta}_{\max}(\tilde{D})\triangleq \max_{\ell}\sqrt{\sum_{d\in\tilde{\mathcal{D}}_\mathsf{imp}}\frac{\mu^2_{\ell,d}}{C_{d,d}}}$. Substituting $\tilde{G}_{\min}={G}_{\min}(\tilde{D})$ and $\tilde{\delta}_{\max}={\delta}_{\max}(\tilde{D})$ into \eqref{eqn: approx_exp_acc} yields the expression of $\hat{A}_{\sf lb}(\mathcal{S},\tilde{D}|\bphi)$.

By combining the results from both random and importance feature ordering cases, we can derive tractable approximations for the expected sensing accuracy, $A(\cS,\tilde{D}|\bphi)$. The main result is presented below, where we define $e_m\triangleq\exp(\phi_m/\tau)$ for notational simplicity and express $\eta$ and $w_m$ in terms of $e_m$.
\begin{MR}
    (Expected Sensing Accuracy Lower Bound). \emph{Given the subset of selected sensors $\mathcal{S}$ and number of selected features $\tilde{D}$, the lower bound of expected sensing accuracy $A(\cS,\tilde{D}|\bphi)$, can be approximated by}
    \begin{equation}\label{eqn: approx_acc}
     \boxed{
    \hat{A}_{\sf lb}(\mathcal{S},\tilde{D}|\bphi) \approx 
    \begin{cases}
        1-(L-1)&Q\left(\frac{\sum_{m\in\mathcal{S}}\sqrt{\frac{\tilde{D}}{D}}e_m \Psi_m}{\sqrt{{\sum_{m\in\mathcal{S}}e_m^2}}}\right), \\ &\quad\quad\,\,\,\text{\emph{random ordering}},\\
        1-(L-1)&Q\left(\frac{\sum_{m\in\mathcal{S}}e_m\Psi_m(\tilde{D})}{\sqrt{{\sum_{m\in\mathcal{S}}e_m^2}}}\right)\,, \\ &\,\,\,\,\,\,\text{\emph{importance ordering}},\\
    \end{cases}
    }
\end{equation}
\emph{where $\Psi_m \triangleq  \frac{\sqrt{{G}_{\min}}}{2}-2{\delta}_{\max}(1-\hat{\pi}_m)$ and $\Psi_m(\tilde{D})\triangleq \frac{\sqrt{{G}_{\min}(\tilde{D})}}{2}-2{\delta}_{\max}(\tilde{D})(1-\hat{\pi}_m)$}.
\end{MR}
We plot the empirical and theoretical lower bounds for both random and importance feature orderings in Fig.~\ref{fig: approx_rnd} and Fig.~\ref{fig: approx_imp}, respectively. The plots demonstrate that the derived lower bounds exhibit a similar trend as the empirical accuracy, with matching optimal numbers of selected sensors. This validates the effectiveness of using $\hat{A}_{\sf lb}(\mathcal{S},\tilde{D}|\bphi)$ as a tractable surrogate for sensor-selection optimization.
\begin{figure}[t]
\centering
\subfigure[Random ordering]{\includegraphics[height=3.4cm]{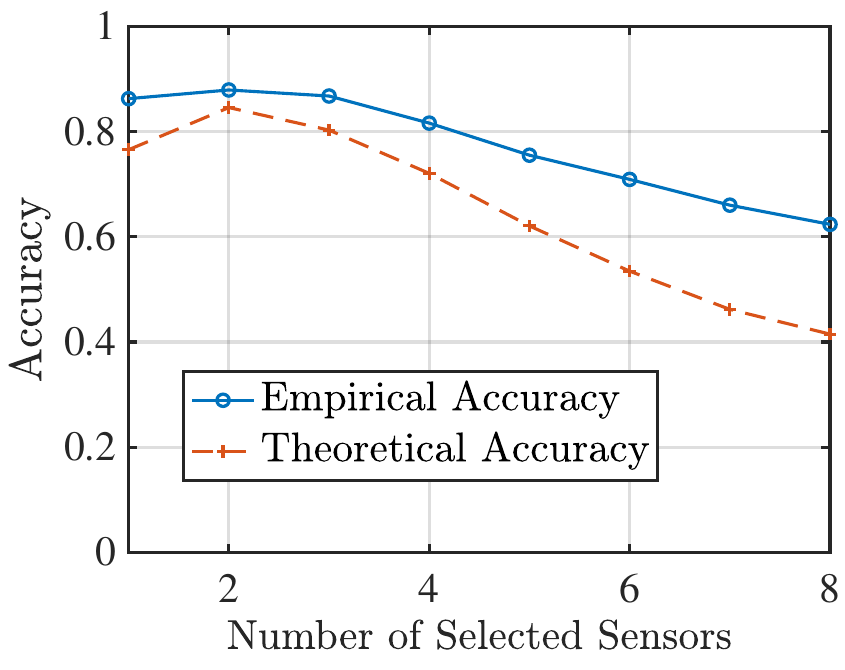}\label{fig: approx_rnd}}
\subfigure[Importance ordering]{\includegraphics[height=3.4cm]{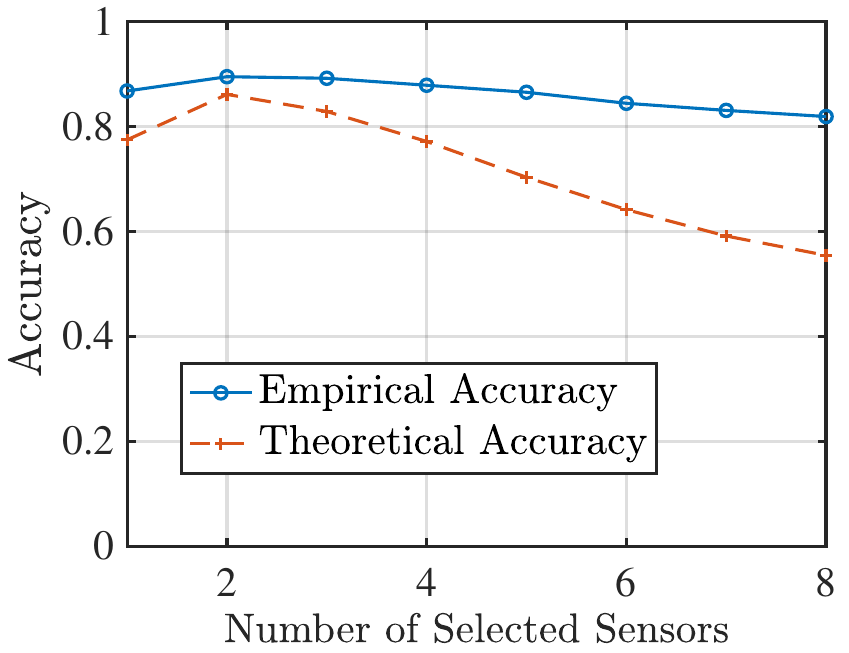}\label{fig: approx_imp}}
\caption{Empirical and theoretical accuracy w.r.t. the number of selected sensors. The sensors are ranked by priority indicators designed shortly in Sec.~\ref{sec: optimal_sel_algo}. The number of classes is set as $L=10$ and the number of features $D=20$.}
\label{fig: mn-max-rician}
\end{figure}
\begin{remark}
    (Expected Classification Margin). \emph{As reflected by Main Result 1, the contribution of each sensor to accuracy is quantified by $\Psi_m$ or $\Psi_m(\tilde{D})$, which can be interpreted as the lower bound of sensor $m$'s expected classification margin, with the expectation taken over random sensor relevance with posterior probability $\hat{\pi}_m$. For example, consider the physical meaning of $\Psi_m \triangleq  \frac{\sqrt{{G}_{\min}}}{2}-2{\delta}_{\max}(1-\hat{\pi}_m)$. Conditioned on $I_m=1$, indicating sensor $m$ is relevant, its feature mean corresponds to the ground-truth class centroid, resulting in a classification margin of $\frac{\sqrt{{G}_{\min}}}{2}$, i.e., the distance from the class centroid to the decision boundary. On the other hand, conditioned on $I_m=0$, the feature mean corresponds to an incorrect class centroid, which, in the worst case, can be on the opposite side of the decision boundary by a distance of $2\delta_{\max}-\frac{\sqrt{{G}_{\min}}}{2}$. Averaging over the two cases with probability $P(I_m=1|\bphi)=\hat{\pi}_m$ yields the expected distance from sensor $m$'s feature mean to the decision boundary, which corresponds exactly to the expression of ${\Psi_m}$. A positive ${\Psi_m}$ indicates that the feature mean is expected to lie on the ``correct'' side of the decision boundary, i.e., closer to the ground-truth class centroid than any other class centroids. }
\end{remark}

\section{Optimal Sensor-Selection Algorithms}\label{sec: optimal_sel_algo}

The sensor-selection problem is aimed at maximizing the expected inference accuracy by selecting participating sensors based on their semantic relevance scores and channel states.  Under the total communication time constraint \eqref{eq:comm_constraint}, the sensor-selection problem is formulated as
\begin{equation*}\text{(P1)}\quad 
        \begin{aligned}
        \max\limits_{\cS,\tilde{D}}\quad&A(\cS,\tilde{D}|\bphi) \\
        \mathrm{s. t.}\quad & 
        \mathcal{S}\subseteq \{1,\ldots,M\},\tilde{{D}}\in \{1,\ldots,D\}, \\
        &\sum_{m\in\mathcal{S}}\frac{q\tilde{D}}{r_m}\leq T.
        \end{aligned}
\end{equation*}
Here, $A(\cS,\tilde{D}|\bphi)$ takes different forms depending on whether random or importance feature ordering is used.
In this section, we develop optimized sensor-selection algorithms by solving Problem P1 with its objective substituted by the approximated lower bound in Main Result 1. The two cases of random and importance feature orderings are discussed respectively in two subsections. For each case, we establish a \emph{priority-based} sensor selection criterion, which ranks sensors in descending order of a priority indicator combining the communication rate and semantic relevance score, and selects top-ranked sensors. Based on this criterion, low-complexity algorithms are further designed to determine the optimal numbers of sensors and features.
\subsection{Case I: Random Ordering}
Using the accuracy surrogate \eqref{eqn: approx_acc} and the monotonicity of Q function, the sensor-selection problem for the random-ordering case is formulated as
\begin{equation*}\text{(P2)}\quad 
        \begin{aligned}
        \max\limits_{\cS,\tilde{D}}\quad&F_{\sf rnd}(\cS,\tilde{D})\triangleq\sqrt{\frac{\tilde{D}}{D\sum_{m\in\mathcal{S}}e_m^2}}\sum_{m\in\mathcal{S}}e_m \Psi_m \\
        \mathrm{s. t.}\quad & 
        \mathcal{S}\subseteq \{1,\ldots,M\},\quad \tilde{{D}}\in \{1,\ldots,D\}, \\
        &\sum_{m\in\mathcal{S}}\frac{q\tilde{D}}{r_m}\leq T.
        \end{aligned}
\end{equation*}
The problem belongs to 0-1 fractional programming, which is in general NP-hard. A standard approach is to perform a continuous relaxation to transform it into a concave-convex fractional programming problem, followed by the Dinkelbach's method\cite{FP}, which iteratively solves a series of convex sub-problems. However, this approach is computationally intensive due to the iterative nature and provides limited insight into the solution structure. Hence, in the subsequent discussion, we pursue a sub-optimal but insightful solution to Problem P2.

\subsubsection{Priority-based Solution Structure}
First, through the following lemma, we show that all sensors included in the optimal $\cS$ shall satisfy $\Psi_m\geq0$. The proof is straightforward and thus omitted for brevity.
\begin{lemma}\label{lemma: positivity_of_expected_DG}
    \emph{Let $(\cS^\dagger,\tilde{D}^\dagger)$ be a feasible solution to Problem P2. Assume that there exists $m\in \cS^\dagger$ such that $\Psi_m < 0$. Then, $(\cS^\dagger\setminus \{m\},\tilde{D}^\dagger)$ is still a feasible solution to Problem P2, but has a larger objective value, i.e., $F_{\sf rnd}(\cS^\dagger\setminus \{m\},\tilde{D}^\dagger)>F_{\sf rnd}(\cS^\dagger,\tilde{D}^\dagger)$.}
\end{lemma}
The lemma shows that removing sensor $m$ with a negative $\Psi_m$ from the selected sensor subset leads to a higher objective value, which implies that the optimal sensor subset should not contain any sensor $m$ with a negative expected classification margin, i.e., $\Psi_m<0$. This aligns with the intuition that including sensors with feature centroids on the incorrect side of the decision boundary will degrade the inference accuracy. Therefore, without loss of generality, we assume that these sensors are excluded from the selection process, such that $\Psi_m\geq 0$ for all $m=1,\ldots,M$.

Next, consider the slave problem of selecting the optimal sensor subset given a fixed number of features $\tilde{D}=\tilde{D}_0$, which can be formulated as
\begin{equation*}\text{(P2.1)}\quad 
        \begin{aligned}
        \max\limits_{\cS}\quad&F_{\sf rnd}(\cS,\tilde{D}_0)\triangleq\sqrt{\frac{\tilde{D}_0}{D\sum_{m\in\mathcal{S}}e_m^2}}\sum_{m\in\mathcal{S}}e_m \Psi_m \\
        \mathrm{s. t.}\quad & 
        \mathcal{S}\subseteq \{1,\ldots,M\},\quad \sum_{m\in\mathcal{S}}\frac{q\tilde{D}_0}{r_m}\leq T.
        \end{aligned}
\end{equation*}
Our approach involves a tight approximation of the objective function $F_{\sf rnd}(\cS,\tilde{D}_0)$, as characterized by the following lemma.
\begin{lemma}\label{lemma: obj_approx_rnd}
    \emph{The optimization objective, $F_{\sf rnd}(\cS,\tilde{D}_0)$, can be bounded as}
    \begin{equation}
        \frac{1}{M}\sqrt{\frac{\tilde{D}_0}{D}\sum_{m\in\cS}\Psi_m^2}\leq F_{\sf rnd}(\cS,\tilde{D}_0)\leq \sqrt{\frac{\tilde{D}_0}{D}\sum_{m\in\cS}\Psi_m^2}.
    \end{equation}
\end{lemma}
\begin{proof} (See Appendix C.)
\end{proof}
While the upper bound is a result of the Cauchy-Schwarz inequality, the lower bound is specific to this problem, relying on the fact that the expected classification margin $\Psi_m$ is an increasing function of the exponential relevance score $e_m$. Lemma~\ref{lemma: obj_approx_rnd} suggests that the objective, $F_{\sf rnd}(\cS,\tilde{D}_0)$, can be approximated by a surrogate $\hat{F}_{\sf rnd}(\cS,\tilde{D}_0)\triangleq \sqrt{\frac{\tilde{D}_0}{D}\sum_{m\in\cS}\Psi_m^2}$. This approximation leads to the following modified slave problem:
\begin{equation*}\text{(P2.2)}\quad 
        \begin{aligned}
        \max\limits_{\cS}\quad &\sum_{m\in\cS}\Psi_m^2 \\
        \mathrm{s. t.}\quad & 
        \mathcal{S}\subseteq \{1,\ldots,M\},\quad \sum_{m\in\mathcal{S}}\frac{q\tilde{D}_0}{r_m}\leq T.
        \end{aligned}
\end{equation*}
This is a standard \emph{knapsack problem} aimed at maximizing the sum \emph{profit} of selected sensors subject to a sum \emph{cost} constraint, where the profit and cost of each sensor are defined as its squared expected classification margin,  $\Psi_m^2$, and required communication time, $q\tilde{D}_0/r_m$, respectively. While the knapsack problem is NP-complete, a well-known near-optimal solution is to incrementally selected sensors based on their \emph{profit densities}, i.e., the profit-to-cost ratios\cite{chen2011applied}. Applying this solution to Problem P2.2 leads to a \emph{priority-based} sensor selection scheme. In this scheme, a priority indicator $\gamma_m$, is defined for each sensor $m$ as follows:
\begin{equation}\label{eqn: priority_rnd}
    \gamma_m = \Psi_m^2 r_m.
\end{equation}
Then, the near-optimal solution is obtained by selecting sensors in descending order of $\gamma_m$ until selecting any additional sensor violates the communication constraint.

\subsubsection{Low-Complexity Algorithm Design}
With the near-optimal solution of the slave problem, Problem P2.2, a straightforward approach for Problem P2.1 is to iterate over $\tilde{D}=1,\ldots,D$, solve Problem P2.2 for every $\tilde{D}$, and selects the one that maximizes the objective function. However, based on \eqref{eqn: priority_rnd}, the priority indicator $\gamma_m$ is invariant to $\tilde{D}_0$, and hence the near-optimal solutions to \emph{all} slave problems follow the same sensor-priority order. The near-optimal sensor subset for the master problem P2.1 then must comprise sensors with top-$S^*$ priorities, as defined in \eqref{eqn: priority_rnd}, for some $S^*=1,\ldots,M$, irrespective of the optimal $\tilde{D}_0$. Since typically $D\gg M$, we propose to directly search for the optimal number of selected sensors, $S^*$, instead of $\tilde{D}_0$ without loss of optimality. The procedure is as follows. For each subset size $S$, select $\cS$ as the sensors with top-$S$ priorities, and choose $\tilde{D}$ as the maximum value to satisfy the communication constraint, i.e., $\tilde{D}=\min\{D,T/\sum_{m\in\cS}q r_M^{-1}\}$. The optimal subset size is thus the one that maximizes objective $F_{\sf rnd}(\cS,\tilde{D})$ across all $S=1,\ldots,M$. The procedure is summarized in Algorithm~\ref{algo: seas_rnd}.

\begin{algorithm}[t] \label{algo: seas_rnd}
\textbf{Input:} Relevance scores $\{\phi_m\}$ and channels $\{h_m\}$;\\
Calculate the rate $r_m$ based on $h_m$ for all $m$\;
Calculate the expected classification margin $\{\Psi_m\}$ based on $\{\phi_m\}_{m=1}^M$ for all sensors\;
Sort sensors in descending order of the priority indicator $\gamma_m=\Psi_m^2 r_m$\;
$F_{\max}\leftarrow 0$\;
\textbf{for} $s=1,\ldots,M$:\\
\ \ \ \ Select $\cS$ as the top-$s$ sorted sensors\;
\ \ \ \ $\tilde{D}\leftarrow \min\{D,T/\sum_{m\in\cS}q r_m^{-1}\}$, $F_s\leftarrow F_{\sf rnd}(\cS,\tilde{D})$\;
\ \ \ \ \textbf{if} $F_s > F_{\max}$ \textbf{then} $F_{\max}\leftarrow F_s$, $\cS^*\leftarrow \cS$, $\tilde{D}^* \leftarrow \tilde{D}$\;
\textbf{Output:} The optimal sensor subset $\cS^*$ and optimal number of features $\tilde{D}^*$.
\parbox{\linewidth}{\caption{Sensor Selection under Random Ordering}}
\end{algorithm}

\subsection{Case II: Importance Ordering}
When the features are selected in the order of importance, the sensor-selection problem is formulated as
\begin{equation*}\text{(P3)}\quad 
        \begin{aligned}
        \max\limits_{\cS,\tilde{D}}\quad&F_{\sf imp}(\cS,\tilde{D})\triangleq\sqrt{\frac{1}{\sum_{m\in\mathcal{S}}e_m^2}}\sum_{m\in\mathcal{S}}e_m\Psi_m(\tilde{D}) \\
        \mathrm{s. t.}\quad & 
        \mathcal{S}\subseteq \{1,\ldots,M\},\quad \tilde{{D}}\in \{1,\ldots,D\}, \\
        &\sum_{m\in\mathcal{S}}\frac{q\tilde{D}}{r_m}\leq T.
        \end{aligned}
\end{equation*}
Similarly, consider the following slave problem given the number of features $\tilde{D}=\tilde{D}_0$:
\begin{equation*}\text{(P3.1)}\quad 
        \begin{aligned}
        \max\limits_{\cS}\quad&F_{\sf imp}(\cS,\tilde{D}_0)\triangleq\sqrt{\frac{1}{\sum_{m\in\mathcal{S}}e_m^2}}\sum_{m\in\mathcal{S}}e_m\Psi_m(\tilde{D}_0) \\
        \mathrm{s. t.}\quad & 
        \mathcal{S}\subseteq \{1,\ldots,M\},\quad \sum_{m\in\mathcal{S}}\frac{q\tilde{D}_0}{r_m}\leq T.
        \end{aligned}
\end{equation*}
It is easy to verify that Lemma~\ref{lemma: positivity_of_expected_DG} also holds for Problem P3.1, which means that any sensor $m$ with $\Psi_m(\tilde{D}_0)<0$ shall not be included in the solution for Problem P3.1. Assuming $\Psi_m(\tilde{D}_0)\geq 0$ for all $m$, we note that Problem P3.1 shares the same structure as Problem P2.1, with the only difference being the replacement of $\sqrt{\frac{\tilde{D}_0}{D}}\Psi_m$ in the former by $\Psi_m(\tilde{D}_0)$ in the latter. This reflects the fundamental difference between random and importance feature ordering: in the random-ordering case, the DG from a subset of features is asymptotically proportional to the number of features $\tilde{D}_0$; in contrast, in the importance-ordering case, the DG is an arbitrary function of $\tilde{D}_0$ depending on the feature importance values. The difference results in a modified priority-based selection scheme for Problem P3.1, where the priority indicator $\gamma_m(\tilde{D}_0)$ for each sensor $m$ is defined as
\begin{equation}\label{eqn: priority_imp}
    \gamma_m(\tilde{D}_0) = \Psi_m^2(\tilde{D}_0) r_m.
\end{equation}
The near-optimal solution to Problem P3.1 is obtained by selecting sensors in descending order of $\gamma_m(\tilde{D}_0)$ until the communication constraint is violated. However, different from the random-ordering case, the priority indicator $\gamma_m(\tilde{D}_0)$ is now dependent on the number of features $\tilde{D}_0$. This is because the deterministic selection of feature dimensions has distinct effects on the expected classification margins of different sensors. Therefore, the near-optimal algorithm for importance-order case nests the priority-based sensor selection into a search for the optimal number of features $\tilde{D}_0$, which is summarized in Algorithm~\ref{algo: seas_imp}.
\begin{algorithm}[t]
\label{algo: seas_imp}
\textbf{Input:} Relevance scores $\{\phi_m\}$ and channels $\{h_m\}$;\\
Calculate the rate $r_m$ based on $h_m$ for all $m$\;
$F_{\max}\leftarrow 0$\;
\textbf{for} $\tilde{D}=1,\ldots,D$:\\
\ \ \ \ Calculate the expected classification margin on the selected feature dimensions $\{\Psi_m(\tilde{D})\}_{m=1}^M$\;
\ \ \ \ Sort sensors in descending order of the priority indicator $\gamma_m(\tilde{D}) = \Psi_m^2(\tilde{D}) r_m$\;
\ \ \ \ $s\leftarrow \max\{s|\sum_{m=1}^s q\tilde{D}r_m^{-1}\leq T\}$\;
\ \ \ \ Select $\cS$ as the top-$s$ sorted sensors\;
\ \ \ \ $\tilde{D}\leftarrow \min\{D,T/\sum_{m\in\cS}q r_M^{-1}\}$,$F_s\leftarrow F_{\sf imp}(\cS,\tilde{D})$\;
\ \ \ \ \textbf{if} $F_s > F_{\max}$ \textbf{then} $F_{\max}\leftarrow F_s$, $\cS^*\leftarrow \cS$, $\tilde{D}^* \leftarrow \tilde{D}$\;
\textbf{Output:} The optimal sensor subset $\cS^*$ and optimal number of features $\tilde{D}^*$.
\parbox{\linewidth}{\caption{Sensor Selection under Importance Ordering}}
\end{algorithm}

\section{Experimental Results}\label{sec: exp_results}

\subsection{Experimental Settings}
We evaluate the sensing-task performance in an ISEA system illustrated in Fig.~\ref{fig: system}. The channels between the server and sensors are assumed to follow i.i.d. Rayleigh fading with a path loss of $-20$ dB. The bandwidth for feature transmission is set to $1$ MHz. The inference performance of all sensor-selection schemes is evaluated on two datasets, which are detailed below.
\begin{itemize}
    \item \textbf{Synthetic dataset:} The synthetic dataset is generated following the GM data model, and the linear classifier \eqref{eq:linear_classifier} is adopted. Unless otherwise specified, the total number of classes is set to $L=40$ and the feature dimension is $D=100$. The GM data statistics are specified as follows. The centroid of each class, $\bmu_\ell$, is i.i.d. randomly sampled from an Euclidean norm ball. The covariance matrix of data noise $\bC$ is a diagonal matrix with each diagonal entry sampled from a uniform distribution between $0$ and $1$. The features of the query image are corrupted by a heavy Gaussian noise with its variance $3$ times larger than that of the sensor observation noise. Each realization includes views of $M=12$ sensors, with the prior probability of each sensor observing the target object set to $\pi_\mathsf{r}=0.4$. The query and keys consist of the first $D_\mathsf{q}=30$ dimensions of the query features and sensor observations, respectively. 
    \item \textbf{ModelNet dataset:} Experiments with non-linear NN classifiers are conducted on the well-known ModelNet dataset \cite{ModelnetPaper}, which provides $12$ distinct views for each object. Each object belongs to one of the $L=40$ classes. To simulate the sensing scenario in Fig.~\ref{fig: system}, the following data-generation procedure is applied. In each realization, a query image of the target object is obtained by the server and corrupted by pixel-wise Gaussian noise. The views of $M=9$ sensors is divided into three groups, each consisting of three views. The first group comprises three different views of the target object of interest to the server, representing views of semantic-relevant sensors. Each of the other groups contains three views of an object with its class different from the ground truth. All sensor views are shuffled such that the server has no positional information on semantic relevance. Following the MVCNN literature\cite{Hang2015ICCV}, a VGG16 model is split into a CNN feature extractor and a fully-connected non-linear classifier. The feature dimension, corresponding to the output layer of the feature extractor, is $512\times 7\times 7=25,088$. Both the query and key encoders comprise sequential fully-connected layers with ReLU activation functions to encode the extracted image features into vectors of length $D_\mathsf{q}=256$, and dot-product attention is applied to calculate the relevance score.
\end{itemize}

The performance of the proposed sensor-selection scheme is compared with four benchmarking schemes, as detailed below.
\begin{itemize}
    \item \textbf{When2com\cite{LYC2020CVPR}:} Using the semantic relevance score $\phi_m$, the server first calculates the normalized weights $w_m$ (see Section~\ref{subsec: feature_prune_fuse}). Then, it selects sensors with normalized weights exceeding $1/M$ to upload features, i.e., $\cS=\{m|w_m>1/M\}$. Attentive fusion is then applied to uploaded features of selected sensors.
    \item \textbf{Best-channel sensor selection:} The sensors are ranked in descending order based on channel gains. Then, a fixed number of sensors with the largest channel gains are selected for feature uploading and attentive fusion. The number of selected sensors is set to $4$ and $3$ for synthetic and ModelNet datasets, respectively.
    \item \textbf{All-inclusive attentive fusion:}
    All sensors are selected for feature uploading and attentive fusion.
    \item \textbf{All-inclusive averaging:} All sensors are selected for feature uploading, and their feature maps are averaged to obtain the fused feature map.
\end{itemize}

\subsection{Performance Evaluation on the Synthetic Dataset}
We first evaluate the E2E inference accuracy performance achieved by different sensor-selection schemes on the synthetic dataset with the linear classifier. The curves of accuracy versus receive SNR levels are plotted in Fig.~\ref{fig_syn} for both random and importance feature ordering. The results show that the proposed semantic-relevance based sensor selection achieves the highest overall accuracy at all SNR levels. Compared with channel-agnostic schemes, i.e., When2com and all-inclusive attentive, the proposed scheme can avoid sensors in deep fades, since the derived priority indicator considers both channels and semantic relevance. Moreover, the numbers of features and semantic-relevant views are optimally balanced to maximize the E2E accuracy. These two advantages result in significant performance gains in low- to moderate-SNR regimes. At high SNR levels where communication resources are sufficient, When2com and all-inclusive attentive also achieve optimal inference accuracy. On the other hand, best-channel selection performs better than channel-agnostic schemes at low SNRs due to the transmission of a considerably larger number of features. However, it fails to converge to optimal accuracy even at high SNRs, as it lacks awareness of semantic relevance and therefore has a fixed probability of including irrelevant views, degrading the accuracy. Comparing Figs.~\ref{fig_syn_rand} and \ref{fig_syn_imp}, one can observe that selecting features in the order of importance yields accuracy gains over random ordering at lower SNRs. However, this gain diminishes as the SNR further increases, as nearly all features can be transmitted in such cases, rendering the impact of feature ordering negligible.

\begin{figure}[t]
\centering
\subfigure[]{\includegraphics[height=3.65cm]{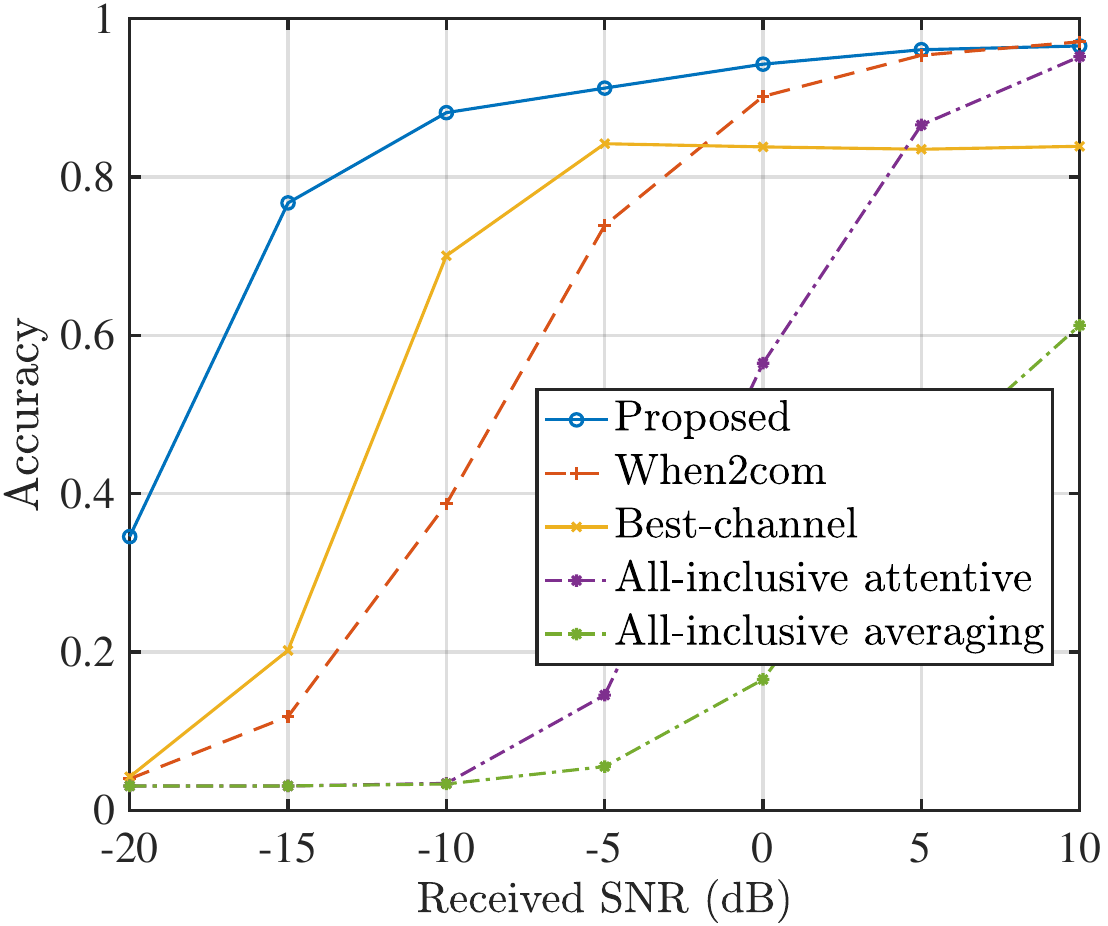}\label{fig_syn_rand}}
\subfigure[]{\includegraphics[height=3.65cm]{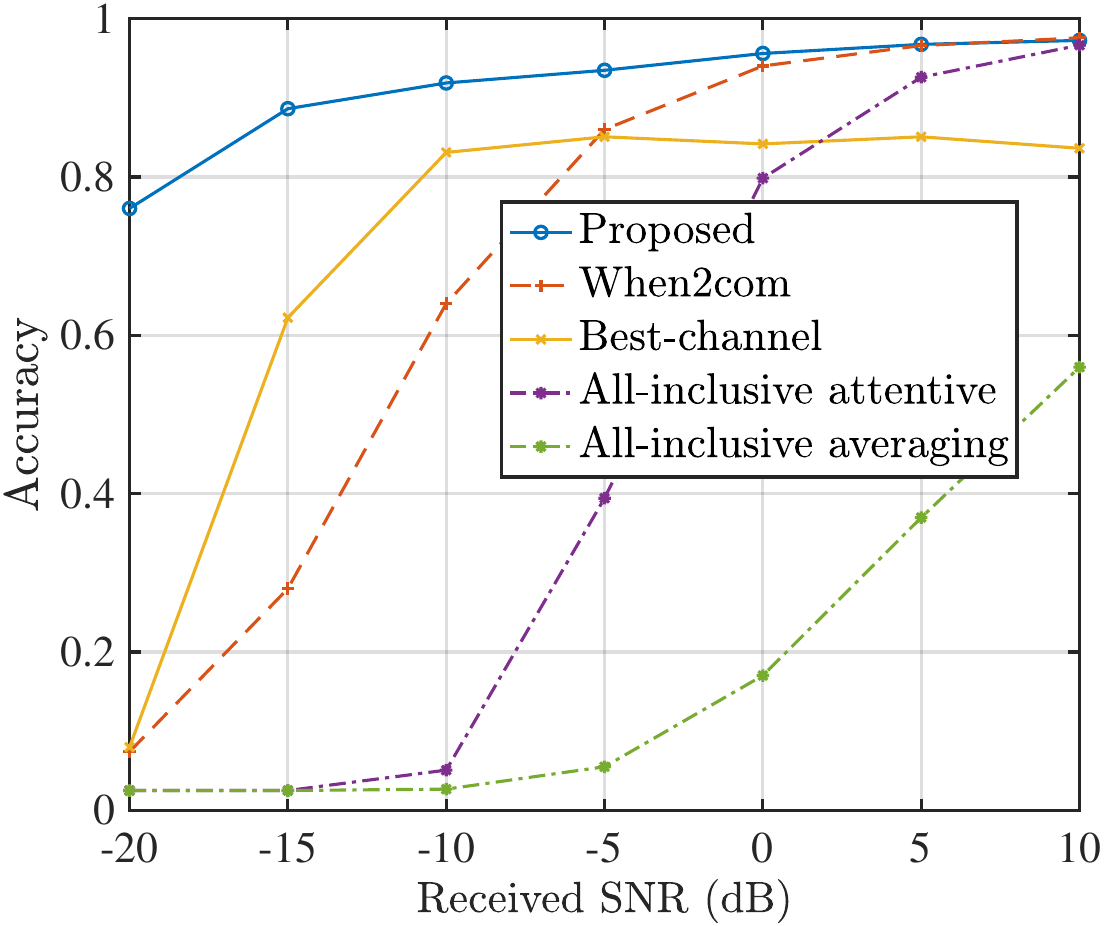}\label{fig_syn_imp}}
\caption{The accuracy performance of sensor selection schemes versus received SNR on the synthetic dataset with (a) random feature ordering and (b) importance feature ordering. }
\label{fig_syn}
\end{figure}

\begin{figure}[t]
\centering
\subfigure[]{\includegraphics[height=3.65cm]{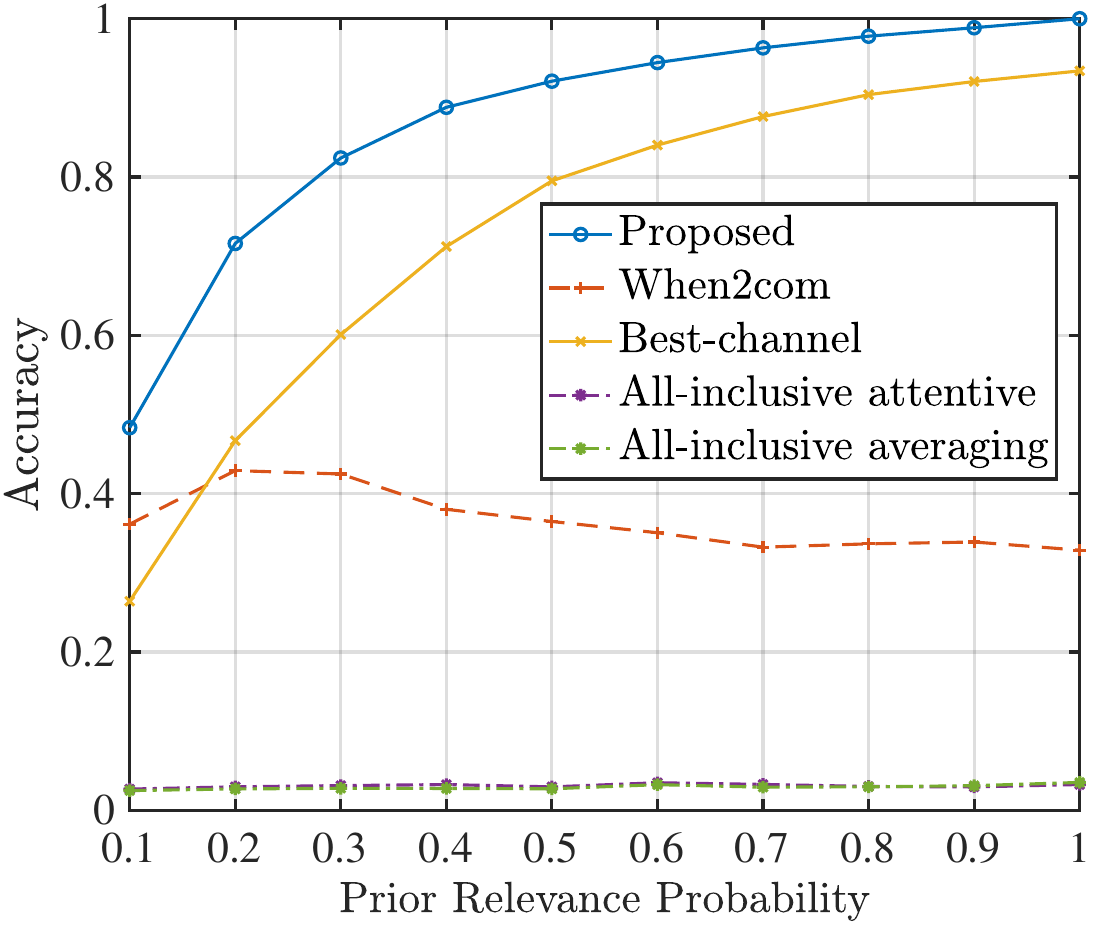}\label{fig_syn_rand_v_pr}}
\subfigure[]{\includegraphics[height=3.65cm]{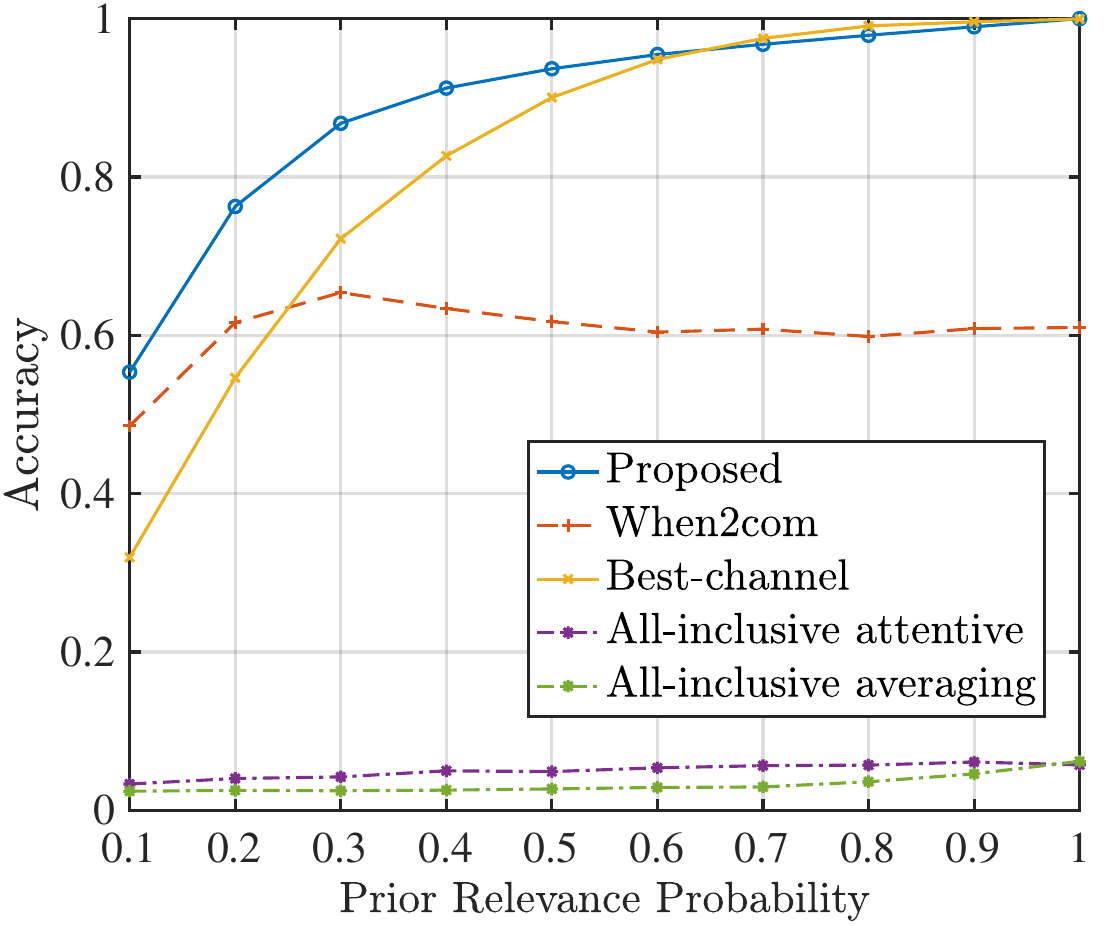}\label{fig_syn_imp_v_pr}}
\caption{The accuracy performance of sensor selection schemes versus received SNR on the synthetic dataset with (a) random feature ordering and (b) importance feature ordering. The received SNR is set as $-10$ dB. }
\label{fig_syn_v_pr}
\end{figure}

Next, we evaluate the accuracy performance as the prior probability of sensor relevance, $\pi_\mathsf{r}$, changes. The results are plotted in Figs.~\ref{fig_syn_rand_v_pr} and \ref{fig_syn_imp_v_pr} for random and importance feature ordering, respectively. The proposed scheme exhibits robustness across different levels of relevance probability as it factors in the prior relevance probability when computing the expected classification margins of sensors. When the relevance probability approaches $1$, best-channel selection demonstrates close performance to the proposed scheme, as in scenarios where all sensors are relevant, channels become the primary factor in sensor selection. Interestingly, despite an initial increase, a gradual decline in accuracy is observed for When2com as the prior relevance probability increases. This trend can be attributed to the following reason: as the relevance probability increases, When2com tends to continuously increase the number of selected sensors. However, due to limited communication resource and the latency constraint, this leads to a significant decrease in the number of uploaded features. The detrimental impact on accuracy outweighs the benefit of including more sensors, leading to an overall performance deterioration.

\subsection{Performance Evaluation on the ModelNet Dataset}

Next, we apply the proposed sensor selection scheme to the ModelNet dataset and compare its performance with benchmarks.
\subsubsection{Sensor selection on real dataset}
We modify Algorithms~\ref{algo: seas_rnd} and \ref{algo: seas_imp} for sensor selection on the ModelNet dataset. The key issue is to compute the expected classification margin for each sensor on the ModelNet dataset, after which the proposed algorithms can be directly applied without further modifications. Consider random ordering first. Recall that the expected classification margin for sensor $m$, $\Psi_m$, is a linear function of the posterior relevance probability $\hat{\pi}_m$: $\Psi_m \triangleq  \frac{\sqrt{{G}_{\min}}}{2}-2{\delta}_{\max}(1-\hat{\pi}_m)$, where $\hat{\pi}_m$ is computed from the relevance score $\phi_m$. However, coefficients of this linear function, which require the minimum pairwise DG, are unknown for MVCNNs. To address this challenge, we propose a simple yet efficient solution through hyperparameter tuning. From Problem P2's objective, we note that scaling $\Psi_m$ by a constant factor all $m$ does not impact the optimal solution of sensor selection. Therefore, without loss of optimality, we normalize $\Psi_m$ as $\Psi_m=1-\lambda(1-\hat{\pi}_m)$, where $\lambda=\frac{4{\delta}_{\max}}{\sqrt{{G}_{\min}}}$. We then conduct a linear search over $\lambda$ as a hyperparameter for maximum accuracy on the training dataset and apply the optimal $\lambda$ for testing. In our experiments, we set $\lambda=1.75$. Next, consider importance ordering, which requires the expected classification margin on a set of top-importance features, i.e., $\Psi_m(\tilde{D})$. We propose an estimation, $\Psi_m(\tilde{D})\approx\theta(\tilde{D})\Psi_m$, where $\theta(\tilde{D})$ is a discount factor as the ratio between sum importance over top-$\tilde{D}$ dimensions and all dimensions, i.e., $\theta(\tilde{D})=\sum_{d\in\cD_{\sf imp}}\bar{g}(d)/\sum_{d=1}^D\bar{g}(d)$. The importance of the $d$-th dimension, $\bar{g}(d)$, is defined as the sum of squared gradients over all weights corresponding $d$-th dimension obtained at the final training round (see\cite{Lan2023TWC} for details).

\begin{figure}[t]
\centering
\subfigure[]{\includegraphics[height=3.65cm]{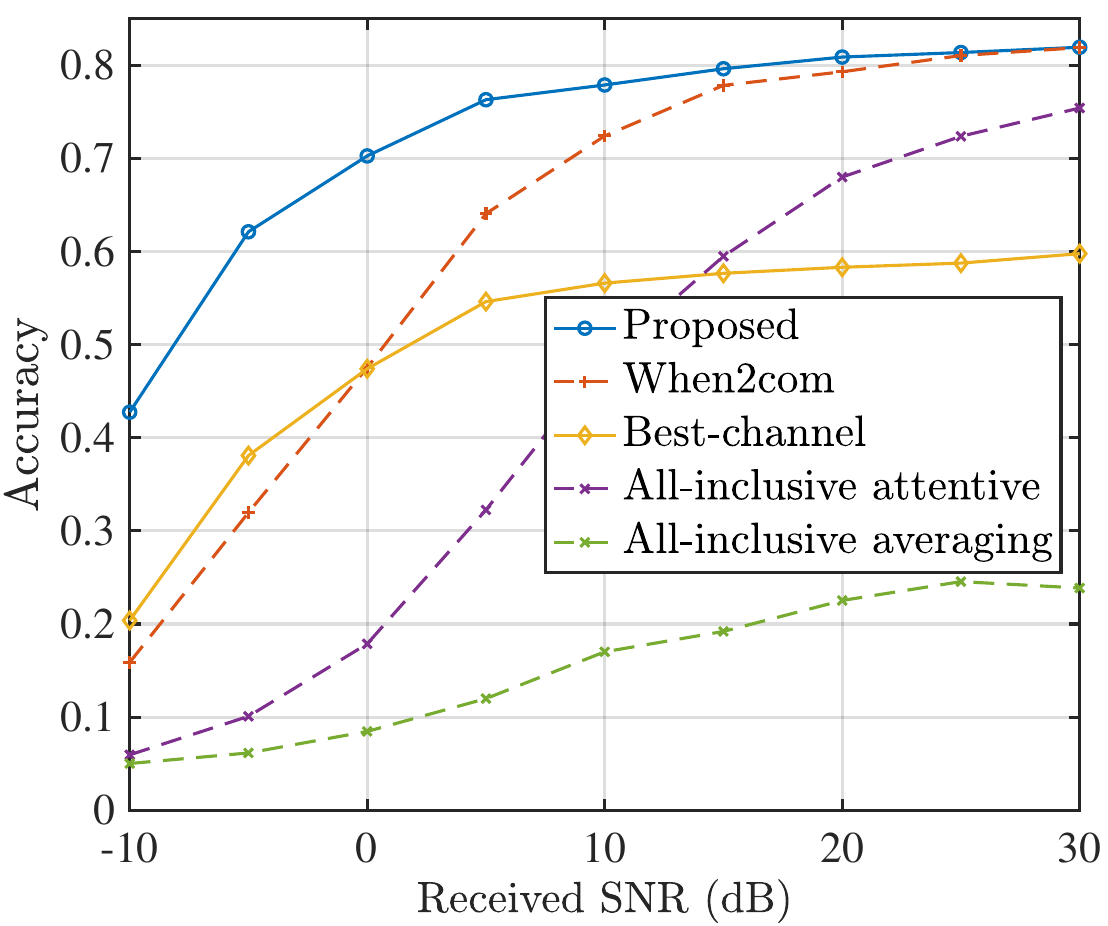}\label{fig_modelnet_rand}}
\subfigure[]{\includegraphics[height=3.65cm]{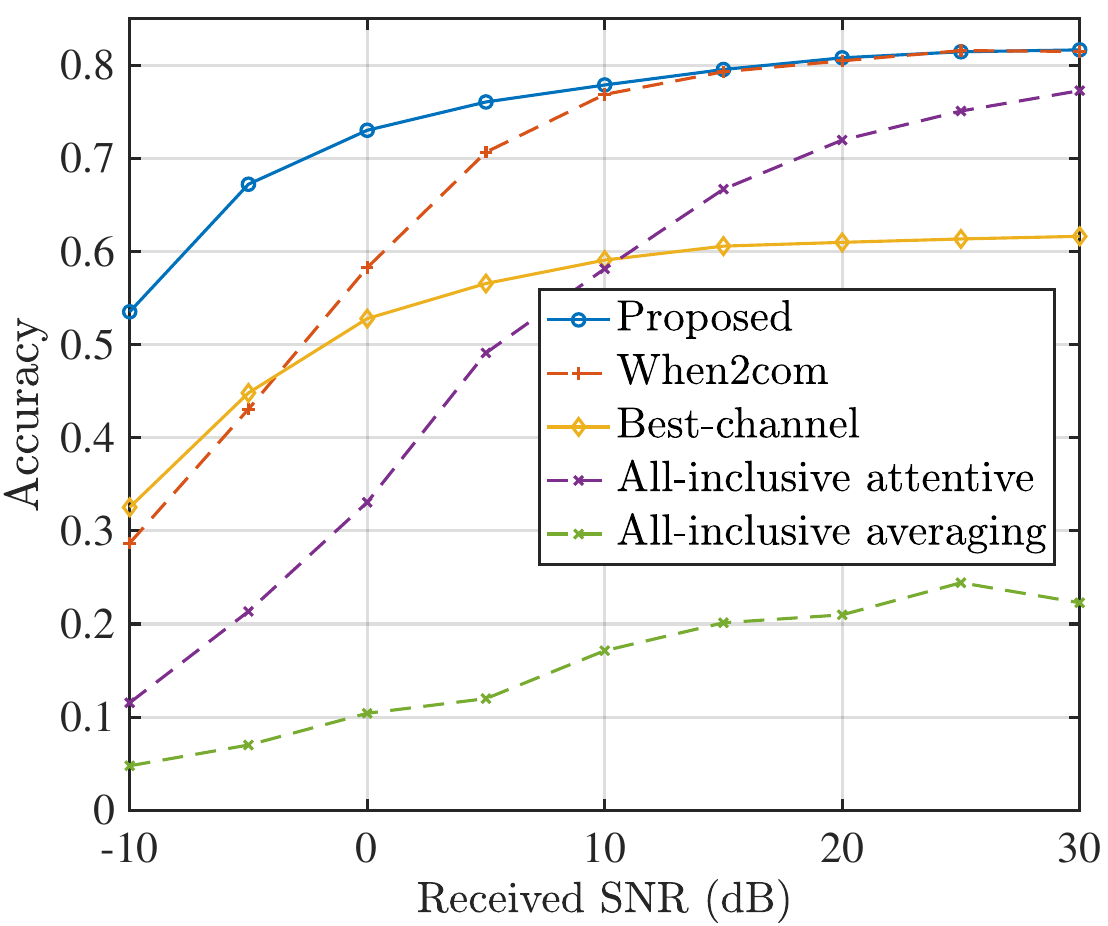}\label{fig_modelnet_imp}}
\caption{The accuracy performance of sensor selection schemes versus received SNR on the ModelNet dataset with (a) random feature ordering and (b) importance feature ordering. }
\label{fig_modelnet}
\end{figure}

\subsubsection{Performance comparisons with benchmarks}
In Fig.~\ref{fig_modelnet}, we plot the curves of accuracy versus receive SNR for random and importance ordering. Similar to the synthetic case, the proposed sensor-selection schemes outperforms all benchmarking schemes on the ModelNet dataset thanks to relevance-aware sensor selection and optimized tradeoffs. Best-channel selection converges at an accuracy of only 60\% as it can miss relevant views even when the communication resource is sufficient. Also, importance feature ordering leads to an accuracy gain over random ordering, particularly at lower SNR levels.

\section{Conclusion} \label{sec: conclusion}
In this paper, we have proposed the framework of semantic-relevance-aware sensor selection to efficiently support distributed ISEA applications. Our E2E accuracy analysis reveals that the semantic relevance scores and channel states of selected sensors jointly determine the expected sensing accuracy. Aimed at accuracy maximization, we have identified a near-optimal sensor-selection policy with a priority-based structure, where sensors are ranked based on a priority indicator that takes into account both relevance scores and channel states. Experimental results show that the proposed sensor-selection algorithms outperform benchmarking schemes significantly.

This work is the first to study the interplay between semantic relevance and channels in ISEA systems, opening up several directions for future research.  One potential direction involves extending the system to a multi-antenna scenario, where random beamforming techniques \cite{dumb_antenna} can be integrated to enhance the sensor-selection gain. Moreover, the development of random access schemes could be pursued to mitigate the signaling overhead associated with centralized sensor selection processes. Another direction is to exploit the temporal correlation of semantic relevance for predictive sensor selection.

\appendix
\subsection{Proof of Theorem 1}
The sensing accuracy conditioned of the relevance of sensors and the observed class of each sensor is given by
\begin{equation}
     A(\cS,\tilde{\cD}|\{I_m\},\{\ell_m\}) = \frac{1}{L}\sum_{\ell_0=1}^L \mathrm{Pr}(\hat{\ell}=\ell_0|\ell_0,\{I_m\},\{\ell_m\}),
\end{equation}
Using the linear classifier defined in \eqref{eq:linear_classifier}, the probability $\mathrm{Pr}(\hat{\ell}=\ell_0|\ell_0,\{I_m\},\{\ell_m\})$ can evaluated and lower bounded using \emph{union bound} as
\begin{align}
    \mathrm{Pr}(\hat{\ell}=\ell_0|\ell_0,&\{I_m\},\{\ell_m\}) = 1-\mathrm{Pr}(\bigcup_{\ell'\neq \ell_0}\{\delta_{\ell_0,\ell'}\geq 0\})\\
    &\geq 1-(L-1) \max_{\ell'\neq\ell_0} \mathrm{Pr}(\delta_{\ell_0,\ell'}\geq 0).\label{eqn: prob_union_bound}
\end{align}
where $\delta_{\ell_0,\ell'} \triangleq z_{\ell_0}(\tilde{\bff}_\mathsf{g}) - z_{\ell'}(\tilde{\bff}_\mathsf{g})$ denotes the Mahalanobis distance from the aggregated feature $\tilde{\bff}_\mathsf{g}$ to $\tilde{\bmu}_{\ell_0}$ subtracted by that to $\tilde{\bmu}_{\ell'}$. Using the definition of Mahalanobis distance, $\delta_{\ell_0,\ell'}$ is further evaluated as
\begin{equation}
    \delta_{\ell_0,\ell'} 
    %& = (\tilde{\bff}_\mathsf{g}-\tilde{\bmu}_{\ell_0})^T \tilde{\mathbf{C}}^{-1}(\tilde{\bff}_\mathsf{g}-\tilde{\bmu}_{\ell_0})-(\tilde{\bff}_\mathsf{g}-\tilde{\bmu}_{\ell'})^T \tilde{\mathbf{C}}^{-1}(\tilde{\bff}_\mathsf{g}-\tilde{\bmu}_{\ell'})\\
     = 2(\tilde{\bmu}_{\ell'}-\tilde{\bmu}_{\ell_0})^T \tilde{\mathbf{C}}^{-1} \tilde{\bff}_\mathsf{g} + \tilde{\bmu}_{\ell_0}^T \tilde{\mathbf{C}}^{-1}\tilde{\bmu}_{\ell_0} - \tilde{\bmu}_{\ell'}^T \tilde{\mathbf{C}}^{-1}\tilde{\bmu}_{\ell'}.
\end{equation}
According to \eqref{eq:cond_agg_feature_dist}, the conditional distribution of $\tilde{\bff}_\mathsf{g}$ is a Gaussian $\mathcal{N}(\rho \tilde{\bmu}_{\ell_0}+\Delta,\eta \tilde{\bC})$, where $\Delta$ is the data component from irrelevant sensors defined as $\Delta \triangleq\sum_{m\in\mathcal{S}, I_m=0}w_m\tilde{\bmu}_{\ell_{m}}$. Then the conditional distribution of $\delta_{\ell_0,\ell'}$ is also Gaussian, $\delta_{\ell_0,\ell'}\sim\mathcal{N}(\bmu_\delta,4\eta G_{\ell_0,\ell'})$, where
\begin{equation}
    \bmu_{\delta
    }\triangleq2(\tilde{\bmu}_{\ell'}-\tilde{\bmu}_{\ell_0})^T \tilde{\mathbf{C}}^{-1}(\rho \tilde{\bmu}_{\ell_0}+\Delta)+\tilde{\bmu}_{\ell_0}^T \tilde{\mathbf{C}}^{-1}\tilde{\bmu}_{\ell_0} - \tilde{\bmu}_{\ell'}^T \tilde{\mathbf{C}}^{-1}\tilde{\bmu}_{\ell'}
\end{equation}
\begin{equation}
    \tilde{G}_{\ell_0,\ell'}\triangleq(\tilde{\bmu}_{\ell'}-\tilde{\bmu}_{\ell_0})^T\tilde{\mathbf{C}}^{-1}(\tilde{\bmu}_{\ell'}-\tilde{\bmu}_{\ell_0}).
\end{equation}
The probability of $\delta_{\ell_0,\ell'}\geq 0$ can then be expressed by the Q function as
\begin{equation}\label{eqn: pairwise_prob_exact}
    \mathrm{Pr}(\delta_{\ell_0,\ell'}\geq 0)=Q\left(-\frac{\bmu_\delta}{2\sqrt{\eta G_{\ell_0,\ell'}}}\right).
\end{equation}
Since $Q(\cdot)$ is a decreasing function, to find an upper bound for $\mathrm{Pr}(\delta_{\ell_0,\ell'})$, we shall obtain an upper bound of $\bmu_\delta$. Via algebraic manipulation, we have
\begin{align}
    \bmu_\delta &= -(\tilde{\bmu}_{\ell'}-\tilde{\bmu}_{\ell_0})^T \tilde{\mathbf{C}}^{-1}\left[\tilde{\bmu}_{\ell'}-(2\rho-1)\tilde{\bmu}_{\ell_0}-2\Delta\right]\\
    &=-\tilde{G}_{\ell_0,\ell'}+2(\tilde{\bmu}_{\ell'}-\tilde{\bmu}_{\ell_0})^T\tilde{\mathbf{C}}^{-1}\left[(\rho-1)\tilde{\bmu}_{\ell_0}+\Delta\right]
\end{align}
Applying \emph{Cauchy-Schwarz inequality} on the second term, we have
\begin{align}
    \bmu_\delta & \leq -\tilde{G}_{\ell_0,\ell'} + 2{\Vert\tilde{\bmu}_{\ell'}-\tilde{\bmu}_{\ell_0}\Vert_{\tilde{\bC}}}\Vert(\rho-1)\tilde{\bmu}_{\ell_0}+\Delta\Vert_{\tilde{\bC}} \\
    & \leq -\tilde{G}_{\ell_0,\ell'} + 2\sqrt{\tilde{G}_{\ell_0,\ell'}}\left[(1-\rho)\Vert\tilde{\bmu}_{\ell_0}\Vert_{\tilde{\bC}}+\Vert\Delta\Vert_{\tilde{\bC}}\right],
\end{align}
where we have used $\sqrt{\tilde{G}_{\ell_0,\ell'}}=\Vert\tilde{\bC}^{-\frac{1}{2}}(\tilde{\bmu}_{\ell'}-\tilde{\bmu}_{\ell_0})\Vert$ and the second inequality is due to the \emph{triangular inequality}. Defining $\tilde{\delta}_{\max}\triangleq \max_{\ell}\Vert\tilde{\mu}_{\ell}\Vert_{\tilde{\mathbf{C}}}$, we further upper bound $\Vert\Delta\Vert_{\tilde{\bC}}$ as 
\begin{align}
    \Vert\Delta\Vert_{\tilde{\bC}} &= \biggl\Vert\sum_{m\in\mathcal{S}}(1-I_m)w_m\tilde{\bmu}_{\ell_{m}}\biggr\Vert_{\tilde{\bC}}\\&\leq\sum_{m\in\mathcal{S}}(1-I_m)w_m\left\Vert\tilde{\bmu}_{\ell_{m}}\right\Vert_{\tilde{\bC}}\\
    &\leq \sum_{m\in\mathcal{S}}(1-I_m)w_m \tilde{\delta}_{\max} = (1-\rho)\tilde{\delta}_{\max}.
\end{align}
Also, noting that $\Vert\tilde{\bmu}_{\ell_0}\Vert_{\tilde{\bC}}\leq\tilde{\delta}_{\max}$, we have
\begin{equation}
    \bmu_\delta \leq -\tilde{G}_{\ell_0,\ell'} + 4(1-\rho)\sqrt{\tilde{G}_{\ell_0,\ell'}}\tilde{\delta}_{\max}.
\end{equation}
Applying the upper bound on $\bmu_\delta$ to \eqref{eqn: pairwise_prob_exact} yields
\begin{equation}
    \mathrm{Pr}(\delta_{\ell_0,\ell'}\geq 0)\leq Q\left[\frac{1}{\sqrt{\eta}}\left(\frac{\sqrt{\tilde{G}_{\ell_0,\ell'}}}{2}-2(1-\rho)\tilde{\delta}_{\max} \right)\right]
\end{equation}
The proof is completed by substituting this upper bound into \eqref{eqn: prob_union_bound} and using $\tilde{G}_{\ell_0,\ell'}\geq \tilde{G}_{\min}$.

\subsection{Proof of Lemma 1}
Using Bayes' theorem, we have
\begin{equation}\label{eqn: posterior_bayes}
\pi_m = \frac{\pi_\mathsf{r}{p}(\phi_m|I_m=1)}{\pi_\mathsf{r}{p}(\phi_m|I_m=1)+(1-\pi_\mathsf{r}){p}(\phi_m|I_m=0)}.
\end{equation}
Conditioned on $I_m=1$, $\bff_m$ is Gaussian distributed as $\bk_m\sim \mathcal{N}(\bmu_0,\bC)$, and hence the conditional distribution of $\phi_m=\bq^T \bW_\mathsf{k}\bff_m$ is ${p}(\phi_m|I_m=1)=\mathcal{N}(\bq^T\bW_\mathsf{k}\bmu_0,\sigma_\mathsf{s}^2)$. Conditioned on $I_m=0$, $\bff_m$ is the mixture of $L-1$ Gaussians, $\mathcal{N}(\bmu_{\ell},\bC)$ for all $\ell\neq\ell_0$. Therefore, we have 
\begin{equation}
    {p}(\phi_m|I_m=0)=\frac{1}{L-1}\sum_{\ell\neq \ell_0}{\mathcal{N}(\phi_m|\bq^T\bW_\mathsf{k}\bmu_{\ell},\sigma_\mathsf{s}^2)}.
\end{equation}
Substituting the conditional distributions of $\phi_m$ into \eqref{eqn: posterior_bayes} yields the expression \eqref{eqn: post_prob}, which completes the proof.
\subsection{Proof of Lemma 4}
The upper bound is trivial by applying Cauchy-Schwarz inequality on $\sum_{m\in\mathcal{S}}e_m \Psi_m$. The lower bound is established by
    \begin{align}
        \sqrt{\frac{1}{\sum_{m\in\mathcal{S}}e_m^2}}\sum_{m\in\mathcal{S}}e_m \Psi_m\geq \sqrt{\frac{1}{\sum_{m\in\mathcal{S}}e_{\max}^2}}e_{\max} \Psi_{\max}\nonumber\\=\sqrt{\frac{1}{|\mathcal{S}|}} \Psi_{\max}\geq \frac{1}{M}\sqrt{\sum_{m\in\mathcal{S}}\Psi_m^2}.
    \end{align}
\bibliographystyle{IEEEtran}
\bibliography{Edge-Inference}

\end{document}

%% file: header.tex
\newtheorem{theorem}{Theorem}

\newtheorem{lemma}{Lemma}

\newtheorem{remark}{\bf Remark}

\def\({\left(}
\def\){\right)}

\def\bee{{\mathbf{e}}}
\def\bff{{\mathbf{f}}}

\def\bk{{\mathbf{k}}}

\def\bq{{\mathbf{q}}}

\def\b0{{\mathbf{0}}}

\def\bC{{\mathbf{C}}}
\def\bD{{\mathbf{D}}}

\def\bW{{\mathbf{W}}}

\def\bmu{{\boldsymbol{\mu}}}
\def\bphi{{\boldsymbol{\phi}}}
\def\cD{\mathcal{D}}

\def\cS{\mathcal{S}}

\usepackage{tabularx}

\newcounter{protocol}